\documentclass{article}
\usepackage{graphicx} 
\usepackage{todonotes}

\usepackage[utf8]{inputenc}
\usepackage[USenglish]{babel}
\usepackage{amsmath,amssymb,amsfonts,amsthm}
\usepackage{mathabx}
\usepackage{mathtools}
\usepackage[T1]{fontenc}
\usepackage{url}
\usepackage{fullpage}
\usepackage{hyperref}
\usepackage{cleveref}
\usepackage[usenames,dvipsnames,table]{xcolor}
\usepackage{booktabs}
\usepackage{authblk}
\usepackage{todonotes}
\usepackage{comment}
\usepackage{paralist}
\usepackage{enumerate}
\usepackage{enumitem}
\usepackage{authblk}
\usepackage{scalerel}
\usepackage{thm-restate}
\usepackage[title]{appendix}
\usepackage{soul}
\usepackage{nicefrac}
\usepackage{thm-restate}
\usepackage{listings}
\usepackage[multiple]{footmisc}
\usepackage{bbm}
\usepackage[linesnumbered,boxruled,vlined, boxed]{algorithm2e}
\usepackage{soul}
\usepackage{thm-restate}
\usepackage{xspace}
\usepackage{empheq}
\usepackage[modulo,right]{lineno}

\title{On Diverse Solutions to Max-$k$-CSP and Bounded Degree $k$-SAT}
\author[1]{Mayank Goswami\footnote{Supported by NSF grant CCF-2503086}}
\author[2]{Adarsh Srinivasan \footnote{Supported by NSF CCF award 2313372 and NSF CAREER award 2443697}}
\affil[1]{Queens College, City University of New York\\
\tt{mayank.goswami@qc.cuny.edu}}
\affil[2]{Rutgers University\\
\tt{adarsh.srinivasan@rutgers.edu}}

\date{\vspace{-8ex}}
\newtheorem{theorem}{Theorem}

\newtheorem{lemma}[theorem]{Lemma}

\newtheorem{corollary}[theorem]{Corollary}

\theoremstyle{definition}
\newtheorem{definition}{Definition}
\newtheorem{remark}{Remark}

\newcommand{\vbl}{\textsf{vbl}}

\newcommand{\z}{\mathbf{z}}
\newcommand{\x}{\mathbf{x}}
\newcommand{\y }{\mathbf{y}}

\newcommand{\f}{\Phi}

\newcommand{\Om}{\Omega_{\f}}

\DeclareMathOperator*{\E}{\mathbb{E}}

\newcommand{\D}{\mathsf{DIAM}}

\newcommand{\Div}{\sigma}

\newcommand{\eps}{\ensuremath{\varepsilon}}

\newcommand{\IR}{\ensuremath{\mathbb{R}}}

\newcommand{\floor}[1]{{\left\lfloor{#1}\right\rfloor}}

\newcommand{\poly}{\mbox{\rm poly}}

\DeclareMathOperator*{\Prob}{\ensuremath{\textnormal{Pr}}}
\renewcommand{\Pr}{\Prob}

\newenvironment{tbox}{\begin{tcolorbox}[
		enlarge top by=5pt,
		enlarge bottom by=5pt,
		 breakable,
		 boxsep=0pt,
                  left=4pt,
                  right=4pt,
                  top=10pt,
                  arc=0pt,
                  boxrule=1pt,toprule=1pt,
                  colback=white
                  ]
	}
{\end{tcolorbox}}

\begin{document}

\maketitle

\begin{abstract}
    We study the problem of generating diverse solutions to Max-$k$-CSP and bounded-degree $k$-SAT, focusing on two distinct metrics: constraint diversity and variable diversity.

For constraint diversity, the goal is to output $s \geq 2$ assignments to the CSP such that each assignment satisfies a $c$-fraction of the constraints, while maximizing the diversity among the $0$-$1$ indicator vectors of satisfied constraints in the Hamming metric. By reducing this to a multi-criteria optimization problem, we design $\poly(n,s)$ time approximation algorithms that return s assignments achieving provable bi-criteria guarantees on both the fraction of satisfied constraints and diversity of the constraint vectors.

For variable diversity, the objective is to maximize the Hamming distance between the assignments, while also maximizing the number of constraints satisfied. For Max-$k$-CSP instances when the desired number of solutions is $s=\exp(n)$, we implicitly represent these diverse approximate solutions by constructing linear codes within the solution space.

Finally, we investigate variable diversity for $k$-SAT in the Lovász Local Lemma regime. In this setting, we establish NP-hardness for the exact diversity problem (computing the diameter of the solution space) and provide a polynomial-time approximation algorithm to efficiently generate diverse satisfying assignments.

\end{abstract}
\newpage

\newpage
\section{Introduction}

The paradigm of generating diverse solutions to various computational tasks has garnered significant attention in recent years~\cite{hebrard2005finding, fomin2020diversity, hanaka2023framework, baste2022diversity, drabik2024finding, iwamasa2025general, galvez2025computing,galvez2025framework}. Broadly speaking, the goal is to output multiple solutions to an instance of a computational problem that are as \emph{diverse} or conceptually different from each other as possible. The practical need to find diverse satisfying assignments has long been recognized in the SAT-solving community, initiated in~\cite{nadel2011generating, agbaria2010sat} using CDCL-based methods, and extensively investigated since~\cite{liang2025diversat,zheng2025exact}.

The recent work of~\cite{austrin2025algorithms} in ICALP'25 designed $O^*(\poly(s)2^{n-cn/k})$-time algorithms\footnote{The $O^*$ notation is used to suppress polynomial factors in $n$.} (for a constant $c$ not depending on $k$) for obtaining $s$ diverse satisfying assignments to \emph{arbitrary} $k$-SAT and $k$-CSP instances, comparing favorably with the state-of-the art algorithms for finding single satisfying assignments~\cite{paturi1997satisfiability,schoning1999probabilistic,paturi2005improved,scheder2024ppsz}. 

What about polynomial time algorithms? Note that the problem of finding two most diverse satisfying assignments is polyAPX-complete even for 2SAT~\cite{crescenzi_hamming_2002}. For what formulae can we then expect polynomial time (perhaps approximation) algorithms for finding $s \geq 2$ diverse solutions? This question is the focus of this paper.

Progress when $s=2$ (diverse-pair) on special classes of formulae has been achieved in ISAAC'24 by \cite{misra2025parameterized} and a preprint by~\cite{gima2024computing}. These works study the diverse-pair problem for affine formulae, 2SAT, hitting formulae and Horn formulae. Apart from some special cases (e.g., when a variable appears in at most 2 clauses), most other cases have NP or W[1]-hardness results, and novel FPT algorithms, but that are still exponential in the worst case. This paper studies three settings.

\begin{enumerate}
    \item \textbf{Diverse Max-SAT:}  Rather than requiring diverse assignments each of which satisfies \emph{all} clauses, the goal is to satisfy a \emph{large fraction} of clauses. While there is prior work on Diverse-Max-SAT~\cite{zhou2023lsdtkms}, ours is the first with theoretical guarantees.
    \item \textbf{Diverse SAT in Local Lemma regime:} We study the diverse SAT problem for formulae that satisfy the conditions of the Lovász local lemma (LLL). These are known to be polynomial time solvable by the celebrated Moser-Tardos algorithm~\cite{moser2010constructive}.
    \item  \textbf{Constraint Diversity:} In all previous work on diverse SAT, diversity of a pair of assignments is naturally defined to be the Hamming distance between them, and this definition then extended for $s>2$ assignments by either considering the minimum or the average pairwise Hamming distance among the $s\choose{2}$ pairs. Motivated by discussions with practitioners in the field of SAT solving, we also consider \emph{constraint diversity}, where the diversity between two assignments is defined as the number of constraints satisfied by exactly one of them. 
\end{enumerate}

\noindent Since our results for the first setting extend to CSPs, we call the first setting \emph{Diverse Max-CSP}. The second setting is termed the \emph{Local Lemma Regime}. These two settings, and the two notions of diversity (the usual, assignment-Hamming distance, and the constraint diversity introduced here) give rise to four questions that we now pose formally.

\noindent
\textbf{Preliminaries.} We start by formally defining the problems we study. Consider a set of $n$ variables $x_1, \dots, x_n$. An assignment $\x \in \{0,1\}^n$ assigns to the truth value $\x_i$ to the variable $x_i$. A constraint $C$ is a function that takes as input assignments to a subset of the variables $\vbl(C)$ and outputs $0$ or $1$. An assignment $\x$ \emph{satisfies} a constraint $C$ if $C(\x)=1$.\footnote{By $C(\x)$, we simply mean $C$ evaluated on the values of $\x$ on $\vbl(C)$.} A $k$-CSP on $n$ variables is a collection $\Phi = (C_1, \dots, C_m)$ of constraints with $|\vbl(C_i)| \le k$ for all $i$. A \emph{weighted $k$-CSP} additionally assigns a weight $w_i \in \mathbb{Z}_{>0}$
to each $C_i$. For $x \in \{0,1\}^n$ we define $|\Phi(x)| = \sum_{i=1}^{m} w_i C_i(x)$ for the total weight satisfied, and set
$\mathsf{CSP}_{\mathrm{OPT}} = \max_{x \in \{0,1\}^n} |\Phi(x)|$. A $k$-SAT instance is a $k$-CSP where each constraint is a disjunction ($\vee$) of at most $k$ variables or negations of variables. An E$k$-SAT instance is a $k$-SAT instance where each constraint depends on exactly $k$ variables.

\subsection{Diverse Max-CSP.} We start by investigating diverse variants of the Max-$k$-SAT and Max-$k$-CSP problems. As natural extensions of satisfiability, these maximization problems have been extensively studied and naturally lend themselves to approximation. Relaxing the objective from exact maximization to approximate maximization serves two purposes. Firstly, all Max-$k$-CSPs admit polynomial time, constant factor approximation algorithms. This gives hope for designing polynomial-time constant-factor approximation algorithms for the \emph{diverse variants of Max-$k$-CSPs} as well\footnote{Note that diverse variants do not always have the same complexity as the original CSP. Finding the two most diverse assignments to a 2-SAT formula is polyAPX-complete~\cite{crescenzi_hamming_2002}}. Secondly, relaxing the strict optimality requirement significantly expands the set of feasible solutions. This allows us to generate a set of solutions that are much more diverse than would be possible when restricted only to exact optima. Indeed, instances with no assignments that satisfy all constraints could have several assignments that satisfy a significant fraction of constraints. We now define our two notions of diversity.

\vspace{3mm}\noindent\textbf{Variable-Diversity versus Constraint-Diversity.} Consider a $k$-CSP instance $\Phi$ with constraints $C_1, \dots, C_m$. For any assignment $\x \in \{0,1\}^n$, we define its corresponding \emph{constraint vector} $\Phi(\x) = (C_1(\x), \dots, C_m(\x)) \in \{0,1\}^m$, where $C_i(\x)=1$ if the constraint $C_i$ is satisfied by $\x$ and $C_i(\x)=0$ otherwise. When searching for a set of $s$ diverse solutions $\x^{(1)}, \dots, \x^{(s)}$, we can evaluate their diversity from two perspectives: either the assignment vectors $\x^{(1)}, \dots, \x^{(s)}$ themselves are diverse in $\{0,1\}^n$ (which we call \emph{variable diversity}), or their respective constraint vectors $\Phi(\x^{(1)}), \dots, \Phi(\x^{(s)})$ are diverse in $\{0,1\}^m$ (which we call \emph{constraint diversity}). Let $d_H$ be the Hamming distance between two strings that counts the number of positions on which they are different. Concretely, we consider the following diversity measures, where a diversity measure $\Div$ is a function that takes as input several strings $\x^{(1)}, \dots, \x^{(s)}$ and outputs a natural number that measures how different they all are. Note that if we required all solutions to satisfy all constraints, every constraint vector would identically be the all-ones vector. Thus, relaxing our requirement to approximately maximize the given instance enables us to define an additional notion of diversity.

\begin{enumerate}
    \item (Variable Diversity, Max-Sum and Max-Min) \sloppy $\sigma_{\text{sum}}(\x^{(1)},\dots, \x^{(s)})=\sum_{i,j \in [s]} d_H(\x^{(i)}, \x^{(j)})$ and $\sigma_{\text{min}}(\x^{(1)},\dots, \x^{(s)})=\min_{i \neq j} d_H(\x^{(i)}, \x^{(j)})$
    \item (Constraint Diversity, Max-Sum and Max-Min) $\sigma_{\text{sum}}(\x^{(1)},\dots, \x^{(s)})=\sum_{i,j \in [s]} d_H(\Phi(\x^{(i)}), \Phi(\x^{(j)}))$ and $\sigma_{\text{min}}(\x^{(1)},\dots, \x^{(s)})=\min_{i \neq j} d_H(\Phi(\x^{(i)}), \Phi(\x^{(j)}))$.
\end{enumerate}

\medskip \noindent
\textbf{Why study constraint diversity?} A compelling motivation for the notion of constraint diversity arises naturally in the context of the \textsc{Max-Cut} problem. Recall that the objective in \textsc{Max-Cut} is to partition the vertices of a given graph $G=(V,E)$ into two disjoint sets, $S$ and $T$, maximizing the number of edges crossing between them. This can be naturally formulated as a Max-$2$-CSP, with the variables indicating which partition a vertex belongs in. However, \emph{variable diversity} does not capture the goal of finding structurally different max-cuts. Indeed, if we seek $s=2$ diverse solutions using variable diversity, this can be achieved trivially by taking any optimal cut and completely swapping the assignments in $S$ and $T$. While this pair of solutions is maximally diverse in the variable space, it represents the exact same cut of $G$. In contrast, the notion of constraint diversity resolves this by measuring the difference in the actual constraints satisfied, making this the `correct' notion of diversity in this scenario. We now informally raise the first question we consider in this work, concerning the constraint diversity concept. 

\begin{quote}
    \textbf{Question 1.} (Constraint Diversity, Max-$k$-CSP) Given a Max-$k$-CSP instance $\Phi$ and $c<1$, can we obtain $s$ assignments $\x^{(1)}, \dots, \x^{(s)}$ each achieving an objective value of at least $c \cdot \mathsf{CSP}_{\text{OPT}}$ (where $\mathsf{CSP}_{\text{OPT}}$ is the optimal value of the CSP), that maximize the $\sigma_{\text{min}}$ or $\sigma_{\text{sum}}$ measures of $\Phi(\x^{(1)}), \dots, \Phi(\x^{(s)})$?
\end{quote}

The above question naturally extends to a weighted $k$-CSP instance, where the assignments must obtain total weight at least a $c$-fraction of the maximum possible attainable weight.

\medskip \noindent
\textbf{Variable Diversity.} The variable diversity question has been previously raised in the context of exact satisfiability problems by~\cite{nadel2011generating,austrin2025algorithms}. The algorithmic problem is that of constructing a set of $s$ \emph{well dispersed} assignments $\x^{(1)}, , \dots, \x^{(s)}$ in $\{0,1\}^n$ that satisfy all the constraints of a given CSP $\Phi$. Let us consider how to extend this question to Max-CSP, where as in Question 1, we want to consider assignments that satisfy at least a $c$-fraction of the constraints. 

When $c=0$, all assignments are $c$-satisfying, and this is simply the well-studied problem of constructing an \emph{error-correcting code} in the Hamming space with asymptotically good rate-distance tradeoffs~\cite{macwilliams1977theory, guruswami2019essential}. Such codes of size $s$ can easily be constructed in $\text{poly}(n,s)$ time. When $c=1$, we are only allowed to return satisfying assignments, and must therefore a) resort to exponential time algorithms, and b) limit ourselves to the maximum diversity available in the space of optimal assignments. 

What about a value of $c$ between 0 and 1? Can we obtain an error-correcting code in the space of $c$-satisfying assignments in polynomial (in $n$,$s$) time? For several well-studied classes of CSPs such as $k$-SAT, the provably optimal algorithm is to simply sample assignments in $\{0,1\}^n$ uniformly at random. We observe that this naturally also outputs assignments that are well-dispersed in the Hamming space (\Cref{thm:random-sampling}). We can go further and ask the same question even when $s$ is large (exponential in $n$). In this setting, one needs a succinct representation of a code. This is possible if the code is \emph{linear}, i.e., forms a vector subspace of $\mathbb{F}_2^n$. This brings us to the following natural question.

\begin{quote}
    \textbf{Question 2.} (Variable Diversity, \textsc{Max-$k$-CSP}) Given a CSP $\Phi$, construct a (linear) code $\mathcal{C}$ such that at least a $1-o(1)$ fraction of the codewords satisfy at least $c \cdot \mathsf{CSP}_{\text{OPT}}$ constraints of $\Phi$.\footnote{The approximation factor $c$ we use varies depending on the type of CSP in question. Typically, we choose it such that $c$-approximation, polynomial time algorithms exist for finding one solution.} Can we do so with the same rate-distance tradeoffs that the usual error-correcting codes attain? 
\end{quote}

By definition, the zero vector $\mathbf{0}$ must be a member of any linear code $\mathcal{C}$, and we can easily construct Max-$k$-CSP instances for which the assignment $\mathbf{0}$ satisfies only a vanishingly small fraction of constraints. Furthermore, the linearity of the code forces the existence of at least some bad codewords with respect to the CSP. Hence, it is not reasonable to expect \emph{every} codeword to satisfy a good fraction of the constraints. This leads us to focus on the question of constructing codes with the vast majority of the codewords achieving a good approximation of the optimum value of the CSP.

\subsection{Diverse-SAT in the local lemma regime.} 

While $k$-SAT is NP-complete, there are several classes of $k$-CNF formulae for which the satisfiability problem can be solved in polynomial time. One of the most well-studied such classes is that of CNF formulae where each clause intersects with a few other clauses. In such cases, the celebrated Lov\'asz Local Lemma implies that the formula is always satisfiable~\cite{erdos1975problems,erdos1991lopsided}. While the original proof of the Lov\'asz Local Lemma was purely existential, a long line of work beginning with Beck~\cite{beck1991algorithmic} and culminating in the algorithm by Moser and Tardos~\cite{moser2010constructive} has made it fully constructive. While the Lov\'asz Local Lemma can be applied to a wide range of CSPs, we focus on the case of E$k$-SAT in this section (recall that E$k$-SAT is the special case of $k$-SAT with each clause of the input formula having exactly $k$ literals). There exists an algorithm that takes as input an E$k$-CNF formula $\Phi$ with $n$ variables and $m$ clauses, such that each clause intersects with at most $\left \lfloor \frac{2^k}{e}-1 \right \rfloor$ other clauses (the ``local lemma regime''), and in expected $\poly(n)$ time outputs a satisfying assignment for $\Phi$. 

Let us consider the special case of $s=2$. Recall that although 2-SAT is in P, finding two most diverse assignments to a 2-CNF formula is polyAPX-complete~\cite{crescenzi_hamming_2002}. For a E$k$-CNF formula $\Phi$, let $\Omega_{\Phi}$ denote the set of assignments that satisfy $\Phi$. We define the \emph{diameter} of $\Phi$ to be $\D(\Phi) = \max_{\alpha,\alpha^{\prime} \in \Omega_{\Phi}} d_{H}(\alpha,\alpha^{\prime})$. 

\begin{quote}
    \textbf{Question 3.}(Diameter, local lemma regime) For any formula $\Phi$ that satisfies the conditions of the Lov\'{a}sz Local Lemma, what is the complexity of computing $\D(\Phi)$?
\end{quote}

Finally, what is the complexity of finding $s>2$ diverse assignments to a $k$-CNF formula in the local lemma regime? 

\begin{quote}
    \textbf{Question 4.} (Variable diversity, local lemma regime) Given an E$k$-CNF formula $\f$ such that each clause intersects with at most $\left \lfloor \frac{2^k}{e}-1 \right \rfloor$ other clauses, and an integer $s=\poly(n)$, can we output $s$ satisfying assignments $\x^{(1)}, \dots, \x^{(s)}$ to $\Phi$ that are maximally diverse in $\poly(n)$ time?  
\end{quote}

\section{Our Results}

We now state our results on the 4 questions in Section~1, starting with Question 1 (Constraint Diversity, Max-SAT). 

\subsection{Constraint Diversity, Max-$k$-CSP}

 Consider $\sigma=\sigma_\text{min}$ or $\sigma_{\text{sum}}$. Assume $\Phi$ is an input $k$-CSP, and let $\mathsf{CSP}_{\text{OPT}}$ be the maximum number of constraints satisfied by any assignment to the variables in $\Phi$. Recall Question 1, where we want to output maximally diverse assignments, each of which satisfies $c\cdot\mathsf{CSP}_{\text{OPT}}$ many constraints. Suppose there exist (optimal) assignments $\x_{\text{OPT}}^{(1)}, \dots, \x_{\text{OPT}}^{(s)}$ achieving a diversity of $\Div(\Phi(\x_{\text{OPT}}^{(1)}), \dots, \Phi(\x_{\text{OPT}}^{(s)})) = \mathsf{D}_{\text{OPT}}$, with each assignment satisfying at least $c \cdot \mathsf{CSP}_{\text{OPT}}$ constraints of the given CSP $\Phi$. Our goal is to output assignments $\x^{(1)}, \dots, \x^{(s)}$ that achieve a diversity of at least $c_1 \mathsf{D}_{\text{OPT}}$, where each assignment satisfies at least $c_2 \cdot (c\cdot \mathsf{CSP}_{\text{OPT}})$ constraints. We refer to this as a $(c_1, c_2)$-approximation. We obtain $(c_1, c_2)$-approximations by designing a reduction to a problem of \emph{simultaneously (approximately) maximizing multiple CSPs}:
\begin{itemize}
    \item \textbf{Max-Sum, Constraint Diversity:} We reduce this problem to the simultaneous maximization of two CSPs (\Cref{thm:max-sum-CSP}). The first is a weighted CSP in which every constraint is either in $\Phi$ or the negation of a constraint in $\Phi$, assigned a positive integer weight of at most $s$. The second is exactly the original CSP.
    \item \textbf{Max-Min Constraint Diversity:} We reduce this variant to the simultaneous maximization of $s+1$ unweighted CSPs (\Cref{thm:max_min_clause_diversity}). In each of these CSPs, the constraints are once again either in $\Phi$ or the negations of constraints in $\Phi$. 
\end{itemize}
Next, we use these reductions to design approximation algorithms for the constraint-diversity problems for various classes of $k$-CSPs. 

\medskip \noindent
\textbf{Max-$k$-XOR using random sampling.} For several well-studied Max-CSPs of interest, such as \textsc{Max-$k$-SAT} and \textsc{Max-$k$-XOR}, the simple strategy of choosing random assignments achieves the optimal approximation ratio for a single instance~\cite{hastad_optimal_2001,hastad_advantage_2004}. Recall that in \textsc{Max-$k$-XOR} each constraint requires the modulo-2 sum (XOR) of exactly $k$ variables to equal a specified parity bit (either $0$ or $1$). In \Cref{cor:xor-maxmin}, we observe that this approach also extends to maximizing $s=\poly(n)$ many simultaneous, unweighted \textsc{Max-$k$-XOR} instances, achieving a Pareto approximation factor of $\frac{1}{2}-\eps$ in polynomial time for every constant $\eps>0$. Consequently, our reduction yields a $(1/2 - \eps,1/2 - \eps)$-approximation algorithm for Max-Min Constraint Diversity in \textsc{Max-$k$-XOR}.

\medskip \noindent
\textbf{General Max-$k$-CSPs.} For general Max-$k$-CSPs, the method of random assignments achieves an approximation ratio of $\frac{1}{2^k}$. However, there exist more involved algorithms: Trevisan developed an algorithm using linear programming which improves this approximation ratio to $\frac{2}{2^k}$~\cite{trevisan1998parallel}, and subsequent work improved this to $\frac{k}{2^k}$~\cite{hast2005approximating,charikar2009near}. This was later shown to be optimal under the unique games conjecture~\cite{samorodnitsky2006gowers}. Bhangale, Kopparty, and Sachdeva~\cite{bhangale2015simultaneous} show that one can achieve an approximation ratio of $\frac{2}{2^k}$ for simultaneously maximizing $\ell$ CSP instances, provided $\ell \leq O(\log^{1/4} n)$. Our reduction then implies that for any constant $\eps>0$, there exists a polynomial time, $(1/2^k - \eps, 1/2^{k-1} - \eps)$-approximation algorithm for Max-Sum Constraint Diversity for $k$-CSPs for $s=\poly(n)$ (\Cref{cor:max_sum_constant}), as well as a $(1/2^k - \eps, 1/2^{k-1} - \eps)$-approximation for Max-Min Constraint Diversity for $k$-CSPs for $s = O(\log^{1/4} n)$ (\Cref{cor:max_min_constant}). 

\medskip \noindent
\textbf{Max-$k$-SAT.} For \textsc{Max-$k$-SAT}, the method of random assignments yields an approximation factor of $1-1/2^k$, which is optimal. Hence, it is natural to ask whether we can get a better bi-approximation algorithm for \textsc{Max-$k$-SAT} than for \textsc{Max-$k$-CSP}. For simultaneous \textsc{$k$-SAT} maximization, there exists a $0.75-\eps$ bi-approximation~\cite[Theorem 1.3]{bhangale2015simultaneous}\footnote{Surprisingly, even for simultaneously maximizing $2^{k-3}$ instances of \textsc{Max-$k$-SAT}, it is NP-hard to get a $\frac{7}{8}+\eps$-Pareto approximation, for any constant $\eps>0$~\cite[Proposition 1.1]{bhangale2015simultaneous}.} . However, we cannot use this algorithm for the Max-Sum Constraint Diversity problem for $k$-SAT. This is because our reduction generates two CSPs. The second CSP is the original $k$-SAT instance. The first instance, however, is a weighted CSP where each constraint is either a conjunction (AND) of $k$ variables or a disjunction (OR) of $k$ variables. Due to the work of Trevisan~\cite{trevisan1998parallel}, it is known that approximately maximizing a CSP where each constraint is a conjunction of $k$ variables is at least as hard as maximizing a general $k$-CSP, and hence, our approximation factors for $k$-SAT are only as good as one would expect for general $k$-CSPs.

\begin{table}[htbp]
    \centering
    \renewcommand{\arraystretch}{1.5}
    \resizebox{\textwidth}{!}{%
    \begin{tabular}{lllll}
        \toprule
        \textbf{Problem} & \textbf{Diversity Variant} & \textbf{Number of Solutions} & \textbf{Approx. Factor $(c_1, c_2)$} & \textbf{Technique} \\
        \midrule
        \textsc{Max-$k$-XOR} & Max-Min & $\poly(n)$ & $(1/2, 1/2)$ & Random Assignments \\
        General \textsc{Max-$k$-CSP} & Max-Sum & $\poly(n)$ & $(1/2^k - \eps, 1/2^{k-1} - \eps)$ & \cite{bhangale2015simultaneous} \\
        General \textsc{Max-$k$-CSP} & Max-Min & $O(\log^{1/4} n)$ & $(1/2^k - \eps, 1/2^{k-1} - \eps)$ & \cite{bhangale2015simultaneous} \\
        \bottomrule
    \end{tabular}%
    }
    \vspace{0.2cm}
    \caption{Summary of bi-approximation guarantees for constraint diversity achieved via our reductions.}
    \label{tab:reduction_consequences}
\end{table}
\noindent
\textbf{Practical Solvers.} Our reductions are also of significant practical interest due to the extensive line of work on developing efficient solvers for simultaneously maximizing multiple CSPs. The underlying reductions are developed in \Cref{sec:clause-diversity}. 

\noindent
\textbf{Extension to non-Boolean domains.} The constraint-diversity reductions extend without essential modification
to CSPs over any fixed finite alphabet $[q]$. Moreover, the multicriteria
approximation algorithm of Bhangale, Kopparty, and
Sachdeva~\cite[Theorem~1.2]{bhangale2015simultaneous} applies to
Max-$k$-CSP over $[q]$, yielding a Pareto approximation factor of
$1/q^{k-1}-\eps$. Consequently, our constraint-diversity results extend to
non-Boolean domains, with the corresponding approximation guarantees
inherited from their multicriteria algorithm.

\subsection{Variable Diversity, Max-$k$-CSP}
In \Cref{sec:vbl-diverse}, we answer Question 2. We demonstrate that for any linear code with rate $\rho$ (yielding $2^{\rho n}$ codewords) and relative minimum distance $\delta$ (meaning any two codewords differ in at least $\delta n$ coordinates), provided the code has a dual distance of $\Omega(\log n)$, almost all of its codewords satisfy a significant fraction of the given constraints (\Cref{thm:vbldiverse-general}). The requirement of a dual distance $\Omega(\log n)$ is a mild structural assumption; it is satisfied with high probability by random linear codes, as well as by several explicit families of error-correcting codes equipped with efficient decoding algorithms (\Cref{rem:good-codes}). For \textsc{Max-$k$-XOR}, we can leverage the symmetric structure of the predicate to ensure that the variance of the number of satisfied constraints under a uniformly random assignment is exactly $m$, yielding strong concentration for any instance (\Cref{cor:xor}). For \textsc{Max-$k$-SAT}, we introduce a structural parameter $\Lambda(\Phi)$ capturing the clause overlap profile of a $k$-CNF formula, and show that concentration holds whenever $\Lambda(\Phi)$ is small relative to $m \cdot 2^k / (\log n)^k$ (\Cref{sec:structured-sat}).

\subsection{Diameter, Local Lemma Regime}

In \Cref{sec:LLL}, we show that finding two satisfying assignments at maximal Hamming distance is equivalent to the Not-All-Equal E$k$-SAT (NAE-E$k$-SAT) problem. By analyzing bounded-degree NAE-E$k$-SAT, we establish a complexity phase transition. Specifically, we define $g(k)$ as the maximum degree $d$ for which every $(k,d)$-formula is unconditionally NAE-satisfiable, and prove that determining NAE-satisfiability for $(k, g(k)+1)$-formulae is NP-hard (\Cref{thm:NAE-jump}). Furthermore, we relate this phase transition threshold to the standard E$k$-SAT threshold $f(k)$, proving that $g(k) \geq \lfloor f(k)/2 \rfloor$. We leave open the question of whether $g(k)=f(k)/2$, asymptotically. 

\subsection{Variable Diversity, Local Lemma Regime}

Finally, we turn our attention to Diverse-SAT in the Lovász Local Lemma regime. While the constructive Moser-Tardos algorithm efficiently finds a single satisfying assignment, we demonstrate in~\Cref{thm:LLL-dispersion} that surprisingly, the same algorithm can be naturally extended to generate a set of diverse satisfying assignments in polynomial time. 

\section{Constraint Diversity} \label{sec:clause-diversity}
The main result of this section is a reduction from constraint diversity to multi-criteria CSP optimization problems. 

\begin{definition}[$(c_1, \dots, c_{\ell})$-Pareto approximation algorithm to $\ell$-criteria Max-$k$-CSP]

    \par \noindent
    \textbf{Input:} $\ell$ weighted CSPs $\Phi_1, \dots, \Phi_\ell$, $\mathsf{OPT}_1, \dots, \mathsf{OPT}_\ell \in \mathbb{Z}_{\geq 0}$.\footnote{We note that $\mathsf{OPT}_i$ is \emph{not} the optimum value of $\Phi_i$. While for each $i$ there exists an assignment $\x^{(i)}$ achieving the optimum value of $\Phi_i$, this does not imply the existence of a single achievement $\x$ simultaneously achieving these values.} \\

    \par \noindent
    \textbf{Output:} Assignment $\x \in \{0,1\}^n$ such that $|\Phi_i(\x)| \geq c_i\cdot\mathsf{OPT}_i$ for each $i \in [\ell]$ if there exists an assignment $\x^\star \in \{0,1\}^n$ such that $|\Phi_i(\x^\star)| \geq \mathsf{OPT}_i$, for each $i \in [\ell]$. 
\end{definition}


\subsection{Reduction for Max-Sum Constraint Diversity}
The main result of this section shows that Pareto approximating Max-Sum constraint diversity on $\Phi$ reduces to finding an approximation for Bicriteria ($\ell =2$) Max-$k$-CSP problem. 

\begin{theorem}\label{thm:max-sum-CSP}
Suppose there exists an algorithm running in $T(n)$ time that computes a $(c_1, c_2)$-Pareto approximation for the Bicriteria Max-$k$-CSP problem. Then, for any given $k$-CSP instance $\Phi$, there exists an algorithm for the Max-Sum constraint diversity problem that outputs $s$ solutions achieving a $(c'_1, c_2)$- approximation, where $ c'_1 = \max\left\{\frac{c_1}{2}, \frac{c_1(s-1)}{(2-c_1)s + c_1}\right\}$. This algorithm runs in $O(s^4 \poly(n) \cdot T(n))$ time. 
\end{theorem}

\begin{remark} \label{rem:same-family-sum}
     The reduction in \Cref{thm:max-sum-CSP} maps a Max-Sum constraint diversity instance $\Phi$ to a Bicriteria Max-$k$-CSP with the following structure: the first CSP has constraints that are either in $\Phi$, or are negations of constraints in $\Phi$, it has positive integer weights bounded by $s$ (or $s$ times the original weights if $\Phi$ was weighted), the second CSP is the original instance $\Phi$. 
\end{remark}

\begin{corollary}\label{cor:max_sum_constant}
Let $k$ be any constant. Let $\mathcal{F}$ be any family of $k$-CSPs. There exists a $\poly(n)$ time, $(1/2^{k}-\eps, 1/2^{k-1}-\eps)$- approximation algorithm for Max-Sum constraint diversity on instances from $\mathcal{F}$, for any $s = \poly(n)$. 
\end{corollary}

\begin{proof}
Using \Cref{thm:max-sum-CSP}, we reduce Max-Sum constraint diversity to an instance of Bicriteria Max-$k$-CSP problem. Then, we use the following theorem from~\cite{bhangale2015simultaneous}, with $\ell=2$. 
\end{proof}

\begin{theorem} \label{thm:bks15} (\cite{bhangale2015simultaneous}, Theorem 1.2)
    For every $\eps>0$, there exists a $2^{O\left(\frac{\ell^4}{\eps^2 \log(\ell/\eps)}\right)}\cdot \poly(n)$ time $\left(\frac{1}{2^{k-1}}-\eps, \dots,\frac{1}{2^{k-1}}-\eps\right)$-Pareto approximation algorithm for $\ell$-criteria Max-$k$-CSP. 
\end{theorem}


To prove \Cref{thm:max-sum-CSP}, we define an intermediate problem: Farthest-Point-Sum-$k$-CSP. Recall that $\Phi$ is an input $k$-CSP, and let $\mathsf{CSP}_{\text{OPT}}$ be the maximum number of constraints satisfied by any assignment to the variables in $\Phi$. Farthest-Point-Sum-$k$-CSP takes as input $t$ elements $\z^{(1)}, \dots, \z^{(t)} \in \{0,1\}^m$. Suppose there exists $\x_{\text{OPT}} \in \{0,1\}^n$ such that $\sum_{i=1}^t d_H(\z^{(i)}, \Phi(\x_{\text{OPT}})) = d^{\star}$ and $|\Phi(\x_{\text{OPT}})| \geq c \cdot \mathsf{CSP}_{\text{OPT}}$. Then, a $(c_1,c_2)$-approximation for Farthest-Point-Sum-$k$-CSP outputs an $\x \in \{0,1\}^n$ such that $\sum_{i=1}^t d_H(\z^{(i)}, \Phi(\x)) \geq c_1 \cdot d^{\star}$ and $|\Phi(\x)| \geq c_2 \cdot (c\cdot\mathsf{CSP}_{\text{OPT}})$. We note the following lemma, which is standard in diverse solutions literature. 

\begin{lemma}[Farthest Point Insertion] \label{lem:farthest-point}
    Suppose that for each $t \leq s$, there exists a $T(n,t)$ for Farthest-Point-Sum-$k$-CSP with approximation factors $(c_1, c_2)$. Then, there exists  $O(s^4 \poly(n) T(n,s))$ for Max-Sum constraint diversity with approximation factors of $c_1', c_2$, where $c_1' = \max \left\{ \frac{c_1}{2},\frac{c_1(s-1)}{(2-c_1)s+c_1} \right\}$. 
\end{lemma}

\label{app:lem:farthest-point}
\begin{proof}
The algorithm has two steps.

\medskip \noindent
\textbf{Step 1: Farthest-point insertion (Gonzalez~\cite{gonzalez1985clustering}).}
Starting from an empty list $L$, we greedily insert points one at a time:
at each step, given the current list $L = \{x^{(1)}, \dots, x^{(t)}\}$, we call
the Farthest-Point-Sum-$k$-CSP algorithm with inputs $z^{(1)} = x^{(1)}, \dots,
z^{(t)} = x^{(t)}$ to obtain a new point $x^{(t+1)}$ maximising
$\sum_{i=1}^{t} d_H(\Phi(x^{(t+1)}), \Phi(x^{(i)}))$ subject to the
constraint-satisfaction guarantee. We repeat until $|L| = s$.

\medskip \noindent
\textbf{Step 2: Local search improvement (Cevallos--Eisenbrand--Zenklusen
\cite{cevallos2016local,cevallos2019improved}).}
After the greedy phase, a local search over single-swap moves
(replace one element of $L$ with the Farthest-Point-Sum-$k$-CSP algorithm's output with the remaining $s-1$ elements) is applied until no swap improves $\sigma_{\mathrm{sum}}$.

\medskip \noindent
\textbf{Approximation factor and time complexity.} Let $\mathrm{OPT}$ denote the optimal $\sigma_{\mathrm{sum}}$ value among all
$s$-tuples whose constraint vectors each satisfy the $c_2$ guarantee. Each call to Farthest-Point-Sum-$k$-CSP
call runs in time $T(n, t) \le T(n, s)$ and the farthest point insertion phase phase makes $s$ calls;
the local search makes at most $O(s^3)$ calls (there are $O(s)$ swaps and each
reduces to a constant number of oracle calls), giving total time
$O(s^4 \operatorname{poly}(n) T(n,s))$. The approximation factor $c_1' = \max\!\left\{\tfrac{c_1}{2},\,
\tfrac{c_1(s-1)}{(2-c_1)s + c_1}\right\}$ arises as follows.
The greedy step alone gives a $c_1/2$-approximation to $\sigma_{\mathrm{sum}}$
(a consequence of the triangle inequality and the $(c_1,c_2)$-bicriteria
guarantee). The local search then improves this to the second term, which
dominates for $s \geq 3$. For the complete
calculation of these factors, and a proof that the local search algorithm terminates in $O(s^3)$ iterations, we refer to~\cite[Lemma~23, Appendix~B]{austrin2024algorithmsarxiv},
whose argument applies without modification to our setting.
\end{proof}

\medskip \noindent
    \textbf{Proof of \Cref{thm:max-sum-CSP}.} To prove~\Cref{thm:max-sum-CSP}, what remains is to reduce Farthest-Point-Sum-$k$-CSP to Bicriteria Max-$k$-CSP instances with the same approximation factors. Consider Farthest-Point-Sum-$k$-CSP with inputs $\z^{(1)}, \dots, \z^{(t)} \in \{0,1\}^m$ and a CSP $\Phi$. Suppose there exists $\x_{\text{OPT}} \in \{0,1\}^n$ such that $\sum_{i=1}^t d_H(\z^{(i)}, \Phi(\x_{\text{OPT}})) = d^{\star}$ and $|\Phi(\x_{\text{OPT}})|\geq c\cdot\mathsf{CSP}_{\text{OPT}}$. Let $\Phi=(C_1, \dots,C_m )$. For each $i \in [m]$, we define $\widebar{C}_i(\x)=C_i(\x ) \oplus 1$, the negation of $C_i$. We now define the weighted CSP $\Psi$ as follows. It has the clauses $C_1, \dots,C_m , \widebar{C}_1, \dots, \widebar{C}_m$. For each $i \in [m]$, the clause $C_i$ has the weight $w_i=|\{j \in [t]: \z^{(j)}_i =0\}|$, and $\widebar{C}_i$ has the weight $\widebar{w}_i=|\{j \in [t]: \z^{(j)}_i =1\}|$. For any $\y \in \{0,1\}^n$, 

\begin{equation*}
    \sum_{j \in [t]}  d_H(\Phi(\y), \z^{(j)}) = \sum_{j \in [t]} \left( \sum_{i \in [m]: \z^{(j)}_i=0} C_i(\y)+\sum_{i \in [m]: \z^{(j)}_i=1} \widebar{C}_i(\y) \right) =  \sum_{i=1}^m w_i C_i(\y) +\sum_{i=1}^m \widebar{w}_i \widebar{C}_i(\y) = |\Psi(\y)| \; .
\end{equation*}

Hence, $\x_{\text{OPT}}$ satisfies $|\Psi(\x_{\text{OPT}})|=d^\star$, and $|\Phi(\x_{\text{OPT}})|\geq c\cdot \mathsf{CSP}_{\text{OPT}}$. This implies that a $(c_1, c_2)$-bi-criteria approximation algorithm for bi-critera \textsc{Max-$k$-CSP} with $\Phi,\Psi, c \cdot \mathsf{CSP}_{\text{OPT}},d^\star$ as input will return $\x \in \{0,1\}^n$ such that $|\Psi(\x)|\geq c_1 \cdot d^\star$, and $|\Phi(\x)|\geq c_2 \cdot (c \cdot \mathsf{CSP}_{\text{OPT}})$. Iterating over all choices of $\mathsf{CSP}_{\text{OPT}},d^\star$ completes the proof of \Cref{thm:max-sum-CSP}.
\subsection{Reduction for Max-Min constraint diversity}

In this section, we design a reduction from Max-Min constraint diversity to $(s+1)$-criteria $k$-CSP maximization. 

\begin{theorem}\label{thm:max_min_clause_diversity}
Suppose there exists an algorithm running in $T(n, s)$ time that computes a $(c_1,\dots, c_1, c_2)$-Pareto approximation for the $(s+1)$-criteria Max-$k$-CSP problem on unweighted instances. Then, for any given $k$-CSP instance $\Phi$, there exists an algorithm for the Max-Min constraint diversity problem that outputs $s$ solutions achieving a $\left(\frac{c_1}{2}, c_2\right)$-approximation. Furthermore, this algorithm runs in $O(s \cdot T(n, s))$ time.
\end{theorem}

\label{app:proof:max_min_clause_diversity}
As in the proof of~\Cref{thm:max-sum-CSP}, we define an intermediate problem-- Farthest-Point-Min-$k$-CSP. It takes as input $t$ elements $\z^{(1)}, \dots, \z^{(t)} \in \{0,1\}^m$. Suppose we are promised that there exists $\x_{\text{OPT}} \in \{0,1\}^n$ such that $ d_H(\z^{(i)}, \Phi(\x_{\text{OPT}})) \geq d^{\star}$ for all $i \in [t]$ and $|\Phi(\x_{\text{OPT}})|\geq c\cdot\mathsf{CSP}_{\text{OPT}}$. The goal is to output $\x \in \{0,1\}^n$ such that $\min_{i=1}^t d_H(\z^{(i)}, \Phi(\x)) \geq c_1 \cdot d^{\star}$ and $|\Phi(\x_{\text{OPT}})| \geq c_2 \cdot c \cdot \mathsf{CSP}_{\text{OPT}}$. An algorithm for this problem implies the existence of an algorithm running in time $O(s \cdot T(n,s))$ for Max-Min constraint diversity with approximation factors of $\left(\frac{1}{2}c_1,c_2\right)$:

    \begin{lemma} \label{lem:gonzales-max-min}
        Suppose that for each $t \leq s$, there exists a $T(n,t)$-time algorithm for Farthest-Point-Min-$k$-CSP with approximation factors $(c_1, c_2)$. Then, there exists  $O(s \cdot T(n,s))$ for Max-Min constraint diversity with approximation factors of $c_1/2, c_2$.
    \end{lemma}
    \begin{proof}
        This follows from the farthest point insertion algorithm due to Gonzalez~\cite{gonzalez1985clustering}. 
    \end{proof}

    \medskip \noindent
    \textbf{Proof of~\Cref{thm:max_min_clause_diversity}.} Consider the problem of solving the Farthest-Point-Min-$k$-CSP for $t$ elements $\mathbf{z}^{(1)}, \dots, \mathbf{z}^{(t)} \in \{0,1\}^m$ (where $t \leq s$). Let $\x_{\text{OPT}} \in \{0,1\}^n$ be an assignment such that $\min_{i=1}^t d_H(\mathbf{z}^{(i)}, \Phi(\x_{\text{OPT}})) = d^\star$ and $|\Phi(\x_{OPT})| \geq c \cdot \mathsf{CSP}_{\text{OPT}}$. For each $j \in [t]$, we define the CSP $\Psi^{(j)}$ as follows. For each $i \in [m]$, we include the clause $C_i$ if $\z_i^{(j)}=0$ and $\widebar{C}_i$ if $\z_i^{(j)}=1$. This implies that $|\Psi^{(j)}(\x_{\text{OPT}})| \geq d^\star$ for each $j \in [t]$. Hence, a $(c_1, \dots, c_{\ell})$-Pareto approximation solver for $t+1$-criteria \textsc{Max-$k$-CSP} with inputs $\Psi^{(1)}, \dots, \Psi^{(t)}, \Phi, d^\star, \dots, d^\star, c \cdot \mathsf{CSP}_{\text{OPT}}$ with approximation factors $(c_1, \dots, c_1, c_2)$ will output $\x \in \{0,1\}^n$ such that $d_H(\Phi(\x), \z^{(j)}) \geq c_1 \cdot d^{\star}$ and $|\Phi(\x)| \geq c_2 \cdot (c\cdot\mathsf{CSP}_{\text{OPT}})$. 

\begin{remark} \label{rem:same-CSP-min}
    The reduction in \Cref{thm:max_min_clause_diversity} maps a Max-Min constraint diversity instance $\Phi$ to an $s+1$-criteria Max-$k$-CSP with the following structure: All the CSPs have constraints that are either in the original CSP $\Phi$, or are negations of constraints in $\Phi$. 
\end{remark}

As earlier, the results of~\cite{bhangale2015simultaneous} imply the following corollary. 

\begin{corollary}\label{cor:max_min_constant}
Let $\mathcal{F}$ be any family of $k$-CSPs. There exists a polynomial-time, $(1/2^{k}-\eps, 1/2^{k-1}-\eps)$-approximation algorithm for Max-Min constraint diversity on instances from $\mathcal{F}$, for any $s \leq O(\log^{1/4}n)$.
\end{corollary} 

\medskip \noindent
\textbf{Applications of our result for specific CSPs.} As an example, we now construct an algorithm for simultaneously maximizing \textsc{Max-$k$-XOR} instances. We note that this predicate family is closed under negations, and hence, the instances of $s+1$-criteria Max-$k$-CSP instances we reduce to are all also \textsc{Max-$k$-XOR} instances.
\begin{corollary} \label{cor:xor-maxmin}
    For every constant $\eps>0$, there exists a randomized, $\poly(n)$ time $\left(\frac{1}{2}-\eps,\frac{1}{2}-\eps \right)$-approximation algorithm for the Max-Min constraint diversity problem with the inputs being unweighted \textsc{Max-$k$-XOR} CSPs and $s=\poly(n)$. 
\end{corollary}

We start by noting the following concentration bound. 
\label{app:proof:cor:xor-maxmin}
\begin{lemma} \label{lem:bonami}
    Let $\Phi$ be an unweighted \textsc{Max-$k$-XOR} instance on $n$ variables and $m$ clauses. Let $\x \sim \{0,1\}^n$ be a uniformly random assignment. Then,

    $$\Pr_{\x \sim \{0,1\}^n} [ |\Phi(\x)| \leq \frac{m}{2} -t\sqrt{m} ] \leq \exp \left(-\frac{k 2^{2/k}}{2e}t^{2/k}\right) \; .$$
\end{lemma}

\begin{proof}
    Recall that each constraint $i \in [m]$ is of the form:
\[ \bigoplus_{j \in S_i} x_j = b_i \; , \]
where $|S_i| \leq k$ and $b_i \in \mathbb{F}_2$. For an assignment $x \in \{0,1\}^n$, we define the corresponding $z \in \{\pm 1\}^n$ as $z_i = (-1)^{x_i}$. It is easy to see that $(-1)^{b_i} \prod_{j \in S_i} z_i=1$ if $\bigoplus_{j \in S_i} x_j = b_i$ and $(-1)$ otherwise. Hence, we can define a degree-$k$ multilinear polynomial $P(z)=\sum_{i=1}^m(-1)^{b_i} \prod_{j \in S_i} z_i$ such that $(P(z)+ m)/2$ is equal to $|\Phi(x)|$. Hence, $\Pr_{\x \sim \{0,1\}^n}[|\Phi(\x)| \leq \frac{m}{2} -t\sqrt{m}]=\Pr_{\z \sim \{\pm 1\}^n}[ P(\z) \leq -2t \sqrt{m}]$. We note that the variance of the random variable $P(\z)$ is $m$. The concentration bound that this probability is at most $\exp\left(-\frac{k}{2e} (2t)^{2/k}\right)$ follows from the Bonami-Beckner hypercontractivity theorem. We refer the reader to~\cite[Theorem 9.23]{odonnell2014analysis} for more details.
\end{proof}

\begin{lemma} \label{lem:xor-conc}
    Let $\Phi_1, \cdots, \Phi_s$ be instances of \textsc{Max-$k$-XOR} on $n$ variables and $m$ constraints. There exists an algorithm running in time $\poly(n)$ that, with probability at least $2/3$, outputs an assignment $\x \in \{0,1\}^n$ such that $|\Phi^{(i)}(\x)| \geq \frac{m}{2}\left(1-O\left(\frac{\log^ks}{\sqrt{m}}  \right)\right)$
\end{lemma}

\begin{proof}
    This can be proved by choosing $t= c \cdot \log^ks \sqrt{m}$ for a sufficiently small constant $c$, applying \Cref{lem:bonami} and the union bound. 
\end{proof}

\Cref{lem:xor-conc} implies that the method of randomly sampling assignments in $\{0,1\}^n$ gives us an approximation algorithm for the Max-Min constraint diversity on unweighted \textsc{Max-$k$-XOR} instances, completing the proof of \Cref{cor:xor-maxmin}. 

\section{Variable Diversity in Max-$k$-CSPs}
In this section, we focus on Question 2 (variable diversity, Max-$k$-CSP). We divide our results into two cases: when $s=\text{poly}(n)$, and we require algorithms to run in $\text{poly}(n,s)$ time, and when $s=\exp(n)$, where we want an implicit representation of the $s$ assignments in $\text{poly}(n)$ time. In \Cref{sec:random-assts}, we demonstrate that for $s=\poly(n)$, the simple method of choosing random assignments to a given $k$-CSP can also be used to obtain variable diverse assignments to the given CSP while also simultaneously maximizing the number of satisfied constraints up to a constant approximation factor. Later, in \Cref{sec:conc-xor}, we address Question 2 for exponential regimes of $s$. For $k$-CSPs that satisfy some structural properties, we obtain linear codes of constant relative rate and constant relative distance, such that a $1-o(1)$ fraction of codewords satisfy at least $\frac{1}{2}-o(1)$ fraction of the constraints. Our code is optimal in the sense that it achieves the Gilbert-Varshamov bound.

\label{sec:vbl-diverse}
\subsection{Randomized algorithms for $s=\text{poly}(n)$}
\label{sec:random-assts}
Consider any CSP instance $\Phi$. It is very well known that the method of repeatedly sampling random assignments $\sim \{0,1\}^n$ is a constant factor approximation algorithm. For the case of Max-$k$-SAT, this yields a $(1-1/2^k)$-approximation algorithm, and for Max-$k$-XOR, a $1/2$-approximation algorithm. Now, we show that this algorithm can also be used to obtain $s$ assignments $\x_1, \dots, \x_s$ to $\Phi$ with the same approximation guarantee, while also satisfying $\sigma_{\text{min} } (\x_1, \dots, \x_s) \geq \frac{n}{2}(1-o(1))$. 

\begin{theorem} \label{thm:random-sampling}
    For all $s \leq \frac{1}{n^2}e^{n^{1/3}/4}$, there exists a randomized algorithm running in time $O(s n 2^k)$ that takes as input a $k$-CSP instance $\Phi$ and with probability at least $2/3$, outputs $\x^{(1)}, \dots, \x^{(s)} \in \{0,1\}^n$ such that $\min_{i \neq j} d_H(\x^{(i)}, \x^{(j)}) \geq \frac{n}{2}(1-o(1))$ and $|\Phi(\x^{(i)})| \geq (1-o(1)) \cdot \E_{\x \sim \{0,1\}^n} |\Phi(\x)|$ for each $i \in [s]$. For Max-$k$-SAT, an algorithm with the same guarantees exists that runs in time $O(sn)$. 
\end{theorem}

\label{app:thm:random-sampling}

\begin{proof}
    Fix $\eps = 1/n$ and $\delta = n^{-1/3}$. The algorithm maintains a list $\mathcal{L}$ of assignments in $\{0,1\}^n$ satisfying $\min_{\x, \y \in \mathcal{L}} d_H(\x,\y) \geq \frac{n}{2}(1-\delta)$. We initialize $\mathcal{L} = \emptyset$ and terminate when $|\mathcal{L}| = s$. At each step, we sample $\x \sim \{0,1\}^n$ uniformly at random. We add $\x$ to $\mathcal{L}$ if $|\Phi(\x)| \geq (1-\eps) \E_{\x} |\Phi(\x)|$ and $d_H(\x, \y) \geq \frac{n}{2}(1-\delta)$ for all $\y \in \mathcal{L}$. Let $E_1$ be the event that $|\Phi(\x)| \geq (1-\eps) \E_{\x} |\Phi(\x)|$, and $E_2$ be the event that $d_H(\x, \y) \geq \frac{n}{2}(1-\delta)$ for all $\y \in \mathcal{L}$.

    \medskip \noindent
    Because $|\Phi(x)| \leq m$, for all $x \in \{0,1\}^n$, $\E[|\Phi(\x)|] \leq \Pr[E_1] \cdot m + (1 - \Pr[E_1])(1-\eps)\E[|\Phi(\x)|]$. Solving for $\Pr[E_1]$ yields $\Pr[E_1] \geq \eps \frac{\E[|\Phi(\x)|]}{m}$. Since $\E[|\Phi(\x)|] \geq m/2^k$, we have $\Pr[E_1] \geq \frac{\eps}{2^k} = \frac{1}{n 2^k}$. For a fixed $\y \in \mathcal{L}$, the Chernoff bound implies $\Pr[d_H(\x, \y) < \frac{n}{2}(1-\delta)] \leq e^{-\delta^2 n / 4}$. By the union bound over the elements already present in $\mathcal{L}$, $\Pr[\neg E_2] \leq s \cdot e^{-\delta^2 n / 4}$. Substituting $\delta = n^{-1/3}$, we obtain $\Pr[\neg E_2] \leq s \cdot e^{-n^{1/3} / 4}$. Because $s \leq \frac{1}{n^2}e^{n^{1/3}/4}$, this term is bounded by $\frac{1}{n^2}$.

    \medskip \noindent
    The probability of success in a single step is $\Pr[E_1 \cap E_2] \geq \Pr[E_1] - \Pr[\neg E_2] \geq \frac{1}{n 2^k} - \frac{1}{n^2} = \Omega(\frac{1}{n 2^k})$. The expected number of samples to find one valid assignment is $O(n 2^k)$. To find $s$ assignments, the expected number of iterations is $O(s n 2^k)$.

    One can remove the $2^k$ term in the running time of this algorithm for \textsc{Max-$k$-SAT}. Here, the expected number of satisfied clauses is much higher: $\E[|\Phi(\x)|] = m(1 - 1/2^k)$. Substituting this into the inequality $\Pr[E_1] \geq \eps \frac{\E[|\Phi(\x)|]}{m}$ from the proof yields $\Pr[E_1] \geq \eps(1 - 1/2^k) = \frac{1 - 1/2^k}{n}$. Since $1 - 1/2^k \geq 1/2$ for all $k \geq 1$, we have $\Pr[E_1] = \Omega(1/n)$. Consequently, the expected number of samples required to find each valid assignment drops to $O(n)$. Thus, the total expected number of iterations to find $s$ assignments becomes $O(sn)$, completely removing the exponential dependence on $k$ present in the $O(s n 2^k)$ bound for general $k$-CSPs.
\end{proof}

\begin{remark}
The corresponding algorithm for finding a single assignment can be
derandomized using the method of conditional expectations. In our setting,
however, the algorithm must simultaneously maintain the approximation
guarantees and the pairwise-distance guarantees for multiple assignments.
It is not clear how to encode all these requirements in a pessimistic
estimator that remains efficiently computable under partial assignments.
We therefore leave the derandomization of this algorithm as an open
question.
\end{remark}

\subsection{Preliminaries on linear codes}
\label{sec:linear-codes}

Before we state our results in Section~\ref{sec:conc-xor}, we need some preliminaries on linear codes.

\medskip \noindent
\textbf{Linear codes.} We define the binary entropy function $H_2(x)=-x \log (x) - (1-x) \log(1-x) $, for each $x \in (0,1)$, and the inverse of this function $H_2^{-1}(z)$ for any $z \in (0,1)$ is defined to be the unique $y \in (0,1/2)$, such that $H_2(y)=z$. A binary-code of rate $\rho$ and relative distance $\delta$ is a subset $\mathcal{C}\subseteq\{0,1\}^n$ such that $|\mathcal{C}|=2^{\rho n}$ and $d_H(x,y) \geq \delta n$ for each $x \neq y \in \mathcal{C}$. The elements of $\mathcal{C}$ are referred to as \emph{codewords}. The celebrated Gilbert-Varshamov (GV) bound for codes states that for any $0 < \eps \leq 1-H_2(\delta)$, there exists a code $\mathcal{C}$ with rate $r \geq 1-H_2(\delta)-\eps$ and relative distance $\delta$. However, when $|\mathcal{C}|$ is very large, we often need a succinct representation of a code. This is possible if the code is \emph{linear}, i.e, when $\mathcal{C}$ forms a vector subspace of $\mathbb{F}_2^n$. In such a case, there exists a $n \times \rho n$ matrix $G$ with entries in $\mathbb{F}_2$ such that $\mathcal{C}=\{G x: x \in \mathbb{F}_2^{\rho n}\}$, and the relative distance can be shown to be $\min_{z \in \mathcal{C}}|z|$, where $|z|=|\{i \in [n]: z_i=1\}|$. The GV bound also proves the existence of linear codes with the same rate and relative distance tradeoffs we mentioned above. We can now rephrase question 2 in a more refined manner. 

\begin{quote}
    Let $\Phi$ be a $k$-CSP on $n$ variables and $m$ constraints such that $\E_{\x \sim \{0,1\}^n} |\Phi(\x)|$ (the expected number of constraints satisfied by a random assignment) is $cm$. Construct a linear code $\mathcal{C}$ with rate $\rho$ and relative distance $H_2^{-1}(1-\rho)$ such that at least a $1-o(1)$ fraction of the codewords satisfy at least $(1-o(1))\E_{\x}|\Phi(\x)|$ constraints.
\end{quote}

\noindent
\textbf{The dual distance.} Let $G \in \mathbb{F}_2^{n \times r}$ be a generating matrix for a code $\mathcal{C}$. Let $G_i$ denote the $i$-th column of $G$ For the code generated by $G$, we define the dual code $\mathcal{C}^{\perp}:=\{\y \in \mathbb{F}_2^n: \langle G_i, \y \rangle=0 \text{, for each } i \in [r] \}$. Suppose $G$ possesses the following structural property: any set of $\leq M$ rows of $G$ are linearly independent over $\mathbb{F}_2$. This is equivalent to the dual code $\mathcal{C}^{\perp}$ having minimum distance at least $M+1$, that is every non-zero $\y \in \mathcal{C}^{\perp}$ satisfies $|y| \geq M+1$. The minimum distance of $\mathcal{C}^{\perp}$ is also referred to as the \emph{dual distance} of $\mathcal{C}$. We use $\delta^{\perp}$ to denote the relative dual distance, i.e, the ratio $(M+1)/n$. It is well known in coding theory literature that for any constant $\eps>0$, a random linear code of rate $\rho$, with probability at least $2/3$, has relative distance $\delta \geq H_2^{-1}(1-\rho) -\eps$ and dual distance $\delta^{\perp} \geq H_2^{-1}(\rho) - \eps$~\cite[Theorem 4.2.1]{guruswami2019essential}. We refer the reader to~\cite{macwilliams1977theory,guruswami2019essential} for more details.

\subsection{Linear Codes for Max-$k$-CSPs}
\label{sec:conc-xor}

\textbf{The bias polynomial of a CSP.} Let $\Phi$ be a Max-$k$-CSP instance, consisting of constraints $\Phi_1, \dots, \Phi_m$. For each assignment $x \in \{0,1\}^n$, define the corresponding assignment $z \in \{\pm 1\}^n$ by $z_i = (-1)^{x_i}$. Each constraint $\Phi_i$ is function that acts on a subset $S_i$ of $k$ variables and outputs $1$ if the constraint is satisfied and $0$ otherwise. Hence, for each $\Phi_i$, we can define a multilinear polynomial $P_i$ that acts on only the variables in $S_i$ and has degree at most $k$ such that $P_i(z)=1$ if $x$ satisfies $\Phi_i$ and $P_i(z)=0$ otherwise. Hence, we can define $P(z)=\sum_{i=1}^m P_i(z)$ that counts the number of constraints satisfied by the assignment $z \in \{\pm 1\}^n$. The function $P :\{\pm 1\}^n \to \IR$ is a Pseudo-Boolean function and hence, has a unique multi-linear expansion. For any $S \subseteq [n]$, we define the Fourier coefficient $\widehat{P}$ to be the coefficient of $\prod_{i \in S}$ in this expansion of $P(z)$. Due to the fact that $P$ has degree at most $k$, $P(z)= \sum_{S \subseteq [n], |S| \leq k} \widehat{P}(S) \prod_{j \in S} z_j$. For a uniformly random assignment $\z \sim \{\pm 1\}^n$, we consider the random variable that counts the number of constraints satisfied by $\mathbf{z}$. We define $\E[\Phi]$ and $\mathsf{Var}[\Phi]$ to be the expectation and variance of this random variable, respectively. It is easy to see that $\E [\Phi]=\widehat{P}(\emptyset)$ and  $\mathsf{Var}[\Phi]=\sum_{S \neq \emptyset} \widehat{P}(S)^2$. We will now prove the following theorem.

\begin{theorem} \label{thm:vbldiverse-general}
    Let $\Phi$ be a \textsc{Max-$k$-CSP} instance with $n$ variables and $m$ constraints. Let $\mathcal{C}$ be any linear code with relative dual distance $\delta^{\perp} \geq \frac{c \ln n}{n}$. The probability that $|\Phi(\mathbf{c})| \leq \E[\Phi]-\frac{e}{4} \sqrt{(c \ln n)^k \mathsf{Var}[\Phi]}$ for a uniformly random $\mathbf{c} \sim \mathcal{C}$ is at most $n^{-c}$.
\end{theorem}
The proof of \Cref{thm:vbldiverse-general} uses moment inequalities, with the dual distance of the code guaranteeing that for a randomly chosen $\mathbf{c} \sim \mathcal{C}$, the random variables $\{\mathbf{c}_i\}_{i=1}^n$, while not independent, are $k$-wise independent, for large enough $k$. 

\label{app:vbldiverse-proof}

\textbf{Some preliminaries.} We say that a sequence of random variables $Y_1, \dots, Y_n$ are $M$-wise independent if for any subset $S \subseteq [n]$ of size at most $M$, the random variables $\{ Y_j \}_{j \in S}$ are mutually independent. 

\begin{lemma} \label{lem:linear-kwise}
    Let $\mathbf{v}^{(1)}, \dots, \mathbf{v}^{(M)}$ be a set of linearly independent vectors over $\mathbb{F}_2^r$. Let $X_1, \dots, X_r$ be uniform, mutually independent random variables over $\mathbb{F}_2$. Then, the random variables $\{Y_{i}=\sum_{j=1}^r \mathbf{v}_j^{(i)} X_j \mod 2\}_{i=1}^M$ are mutually independent. 
\end{lemma}

\begin{proof}
    Let $V$ be an $M \times r$ matrix where each row $i$ corresponds to the vector $v^{(i)}$. Since the $M$ vectors are linearly independent, the matrix $V$ has full row rank $M$. The random variables $Y_1, \dots, Y_M$ can be represented as the column vector $Y = VX$, where $X$ is the column vector of the uniform random variables $X_1, \dots, X_r$.

    Because $V$ has full row rank, the linear transformation defined by $V$ is surjective. By the rank-nullity theorem, the kernel of $V$ has dimension $r - M$. Therefore, for every possible vector $y \in \mathbb{F}_2^M$, the linear system $V x = y$ has exactly $2^{r-M}$ solutions for $x$. Since $X$ is drawn uniformly at random from the $2^r$ possible vectors in $\mathbb{F}_2^r$, the probability of obtaining any specific $y \in \mathbb{F}_2^M$ is exactly:
    \[
    \Pr(Y = y) = \frac{2^{r-M}}{2^r} = \frac{1}{2^M}
    \]

    Because every joint outcome $y \in \mathbb{F}_2^M$ occurs with the exact same probability $2^{-M}$, the random vector $Y$ is uniformly distributed over $\mathbb{F}_2^M$. This uniformity over the product space directly implies that its individual components $Y_1, \dots, Y_M$ are mutually independent and identically distributed as uniform random variables over $\mathbb{F}_2$.
\end{proof}




\begin{lemma} \label{lem:dual-indep}
    Let $\mathcal{C}$ be a linear code with dual distance at least $M+1$. Let $Y=(Y_1, \dots, Y_n) \sim \mathcal{C}$ be a uniformly random codeword. Then, the random variables $Y_1, \dots, Y_n$ are $M$-wise independent. 
\end{lemma}

\begin{proof}
    Let $X \in \mathbb{F}_2^r$ be a vector of $r$ independent, uniform random variables $X_1, \dots, X_r$. The random variable $Y = GX$ represents a uniformly random element of the code $\mathcal{C}$. Now, consider any set of $m \leq M$ rows $G_{i_1}, \dots, G_{i_m}$. Suppose these rows were linearly dependent. That is, $G_{i_1}+ \dots+ G_{i_m} = \mathbf{0}$. This would imply that vector $\mathbf{z} \in \mathbb{F}_2^n$ such that $\mathbf{z}_{i_1}=\cdots =\mathbf{z}_{i_m}=1$, and $0$ elsewhere is a member of the dual code $\mathcal{C}^\perp$. However, as $|\mathbf{z}| \leq M$, this violates the dual distance assumption. 
\end{proof}




Finally, we need the following lemma from \cite[Theorem 9.21]{odonnell2014analysis}.  

\begin{lemma} \label{lem:hypercontractivity}
    Let $Q:\{\pm 1\}^n \to \mathbb{R}$ be a multilinear polynomial of degree $k$. For any $d$, we define $\Vert Q \Vert_{d}=\left(\E_{X \sim \{\pm 1\}^n} |Q(X)|^d\right)^{1/d}$. Then, $\Vert Q \Vert_d \leq (d-1)^{k/2} \Vert Q\Vert_2$, for any $k \geq 2$. 
\end{lemma}

\medskip \noindent
\textbf{Proof of \Cref{thm:vbldiverse-general}.} We can use standard concentration bounds to showed that a the number of constraints of $\Phi$ satisfied by a uniformly random $\x \sim \{0,1\}^n$ is well-concentrated around $\E[\Phi]$, using tail bounds on multivariate polynomials over independent random variables. In this setting, we want to show that for a uniformly random $Y=(Y_1, \dots,Y_n) \sim \mathcal{C}$, the number of constraints is concentrated. The random variable $Y=(Y_1, \dots,Y_n) \sim \mathcal{C}$ is obtained by choosing a uniformly random codeword $\mathbf{c} \sim \mathcal{C}$, and mapping it to $\{\pm 1\}^n$. As $Y_1, \dots,Y_n$ are not mutually independent, we cannot use these concentration bounds directly. However, we can use \Cref{lem:dual-indep} which implies that $Y_1, \dots,Y_n$ are $\floor{(c \log n)}$-wise independent. 

Hence, we can bound the higher moments of $P(Y)-\E [\Phi]=P(Y)-P(\emptyset)$. When we compute the expectation of any polynomial function $g$ of the random variables $Y_1, \dots, Y_n$, we can treat them as mutually independent as long as no more than $\floor{(c \log n)}$ of them appear in a product. Because $P$ has degree at most $k$, this implies that the $d$-th moments of $P(Y)-P(\emptyset)$ and $P(Z)-P(\emptyset)$, where $Z $ is uniformly random over $\{\pm 1\}^n$ are equal, as long as $d \leq \floor{(c \log n)}$. For any even $d \leq \floor{(c \log n)}$, we now use \Cref{lem:hypercontractivity} and show that $\left(\E [(P(Z)-P(\emptyset))^d]\right)^{1/d} \leq (d-1)^{k/2} \sqrt{\mathsf{Var}[P]}$. Hence, $ \E [(P(Z)-P(\emptyset))^d] \leq (d-1)^{kd/2} (\mathsf{Var}[P])^{d/2}$. Next, to upper bound the probability that $P(Y)<P(\emptyset)-2t \sqrt{\mathsf{Var}[\Phi]}$, we can use Markov's inequality\footnote{As $d$ is even, if $P(Y)<P(\emptyset)-2t \sqrt{\mathsf{Var}[\Phi]}$, this implies that $(P(Y)-P(\emptyset))^d >\left(2t \sqrt{\mathsf{Var}[\Phi]}\right)^d$, and Markov's inequality states that for any random variable $Z$, $\Pr[Z \geq t] \leq \frac{\E[Z]}{t}$. } to upper bound this by $\frac{\E_{Y}[(P(Y)-P(\emptyset))^d]}{(2t \sqrt{\mathsf{Var}[\Phi]})^d}$. Now, as long as $d \leq \floor{(c \log n)}$, this is at most $\left( \frac{(d-1)^{k/2}}{2t} \right)^d \leq \left( \frac{d^{k/2}}{2t} \right)^d$. Choose $d=(2t/e)^{2/k}$, this is upper bounded by $e^{-(2t/e)^{2/k}}$. The value of $t$ in this theorem is $\frac{e}{2}\sqrt{ (c\ln n)^k}$. Hence, $d$ is at most $c \ln n$, and the probability can be upper bounded by $ n^{-c}$.

\begin{remark}[Codes that satisfy the conditions of \Cref{thm:vbldiverse-general}] \label{rem:good-codes}
     For any constant $\rho \in (0,1)$, the random linear code of rate $\rho$ satisfies the conditions of \Cref{thm:vbldiverse-general}~\cite[Theorem 4.2.1]{guruswami2019essential}. It also has a constant relative distance. There also exist explicit families of error correcting codes with constant relative rate and constant relative distance satisfying the conditions. Examples include Justesen codes~\cite{justesen2003class}.  
\end{remark}

\begin{corollary} \label{cor:xor}
    Let $\Phi$ be an unweighted \textsc{Max-$k$-XOR} instance with $n$ variables and $m$ constraints. Let $\mathcal{C}$ be any linear code with relative dual distance $\delta^{\perp} \geq \frac{c \ln n}{n}$. The probability that $|\Phi(\mathbf{c})| \leq \frac{m}{2}\left(1-\frac{e}{2}\sqrt{\frac{ (c \ln n)^k }{m}}\right)$ for a uniformly random $\mathbf{c} \sim \mathcal{C}$ is at most $n^{-c}$.
\end{corollary} 

\begin{proof}
    As $\Phi$ is a \textsc{Max-$k$-XOR} instance, it is easy to see that $\E[\Phi]=m/2$ and $\mathsf{Var}[\Phi]=m$, and the result follows from~\Cref{thm:vbldiverse-general}. 
\end{proof}

We can prove a weaker result for \textsc{Max-$k$-SAT}. For a $k$-CNF formula $\Phi$ on $n$ variables and $m$ clauses, assuming the clauses are distinct, we can show that $\mathsf{Var}[\Phi] \leq \frac{m^2}{2^k}\left(1-\frac{1}{2^k}\right)$. However, this is not enough to obtain good enough concentration bounds from~\Cref{thm:vbldiverse-general}. Hence, we aim to handle `nice' $k$-CNF formulae, and define a combinatorial parameter that captures this. 

\subsection{Variable diversity in for structured $k$-SAT instances}
\label{sec:structured-sat}
As noted in \Cref{sec:vbl-diverse}, we aim to define a combinatorial parameter that captures formulas admitting good concentration bounds. 
\begin{definition}
    For a $k$-CNF formula $\Phi$ with clauses $C_1, \dots, C_m$ on the sets $S_1, \dots, S_m$, we define the overlap to be 
    
    $$\Lambda(\Phi):=\frac{1}{m 2^k} \sum_{j , j' \in [m]}2^{|S_j \cap S_{j'}|} \; .$$
\end{definition}

\begin{corollary} \label{cor:sat}
    Let $\Phi$ be a \textsc{Max-$k$-SAT} instance on $n$ variables, $m$ clauses and overlap $\Lambda(\Phi)$. Let $\mathcal{C}$ be any linear code with relative dual distance $\delta^{\perp} \geq \frac{c \ln n}{n}$. Then, for a uniformly random codeword $\mathbf{c} \sim \mathcal{C}$, 

    $$ \Pr\left[ |\Phi(\mathbf{c})| \leq m\left(1-\frac{1}{2^k}\right)-\frac{e}{4} \sqrt{\frac{m \cdot \Lambda(\Phi)}{2^k}\cdot (c \ln n)^k} \right] \leq n^{-c} \; .$$
\end{corollary}

\label{app:k-sat-cor}
\begin{proof}
It is enough to show that for any $k$-CNF formula $\Phi$, $\mathsf{Var}[\Phi] \le \frac{m \cdot \Lambda(\Phi)}{2^k}$; the corollary then follows from \Cref{thm:vbldiverse-general}.

Let $C_1, \ldots, C_m$ be the clauses of $\Phi$, with $C_j$ on variable set $S_j$ ($|S_j| = k$) and literal signs $\epsilon_{j,i} \in \{-1,+1\}$ for $i \in S_j$. The satisfaction polynomial for $C_j$ (that indicates whether it is satisfied or not) is
\[
P_j(z) \;=\; 1 \;-\; \prod_{i \in S_j} \frac{1 - \epsilon_{j,i} z_i}{2},
\]
since the product equals $1$ exactly when every literal of $C_j$ is false. Expanding and summing over $j$, the bias polynomial $P(z) = \sum_j P_j(z)$ has Fourier coefficients
\[
\widehat{P}(T) \;=\; \frac{(-1)^{|T|+1}}{2^k} \sum_{j:\, T \subseteq S_j} \prod_{i \in T} \epsilon_{j,i}, \qquad \emptyset \neq T,\ |T| \le k.
\]

By Parseval's theorem,
\[
\mathsf{Var}[\Phi] \;=\; \sum_{\emptyset \neq T} \widehat{P}(T)^2 \;\le\; \frac{1}{4^k} \sum_{\emptyset \neq T} a_T^2,
\]
where $a_T = |\{j : T \subseteq S_j\}|$, using $\bigl(\sum_{j} \pm 1\bigr)^2 \le a_T^2$. Expanding $a_T^2 = |\{(j,j') : T \subseteq S_j \cap S_{j'}\}|$ and switching the order of summation,
\[
\sum_{\emptyset \neq T} a_T^2 \;=\; \sum_{j,j' = 1}^m \bigl(2^{|S_j \cap S_{j'}|} - 1\bigr) \;\le\; \sum_{j,j'} 2^{|S_j \cap S_{j'}|} \;=\; m \cdot 2^k \cdot \Lambda(\Phi).
\]
Dividing by $4^k$ yields $\mathsf{Var}[\Phi] \le m \cdot \Lambda(\Phi) / 2^k$.
\end{proof}

\section{Diverse-SAT in the Local Lemma Regime}
\label{sec:LLL}

In this section, we study variable diversity on input formulae that are guaranteed to be satisfiable via the Lov\'asz Local Lemma (LLL). In \Cref{sec:lll-lower}, we discuss hardness results for the diameter problem $(s=2)$.  In \Cref{sec:alg-disp-lll}, we design a polynomial time algorithm that outputs $s$ satisfying assignments to a given E$k$-CNF formula satisfying the LLL conditions, maximizing the Max-Min diversity up to an approximation factor of $\frac{1}{2}$.

\subsection{Diameter of E$k$-SAT in LLL regime: Enter NAE-E$k$-SAT}
\label{sec:lll-lower}

In this section, we demonstrate how the Diverse-E$k$-SAT problem for a bounded-degree formula is related to the NAE-E$k$-SAT problem for the same formula. An instance of NAE-E$k$-SAT is a E$k$-CNF formula $\Phi$, and it is NAE-satisfied by an assignment $\alpha$ if, for every clause $C$ of $\Phi$, $\alpha$ sets at least one literal in $C$ to true and at least one literal to false. Let $\Omega_{\Phi}$ denote the set of assignments that satisfy $\Phi$. We define the diameter of $\Phi$ to be $\D(\Phi) = \max_{\alpha,\alpha^{\prime} \in \Omega_{\Phi}} d_{H}(\alpha,\alpha^{\prime})$. This is a special case of the Diverse-E$k$-SAT problem for $s=2$. A formula $\Phi$ is NAE-satisfiable if and only if there exists an assignment $\alpha$ such that both $\alpha$ and $\bar{\alpha}$ (the assignment obtained by flipping every bit of $\alpha$) satisfy $\Phi$. Hence, a formula $\Phi$ has diameter $n$ if and only if it is NAE-satisfiable.

\begin{definition}
A $(k,d)$-formula is a CNF formula where each clause has width exactly $k$, and each variable appears in at most $d$ clauses. $(k,d)$-SAT is the computational problem of determining whether a given $(k,d)$-formula is satisfiable. Similarly, NAE-$(k,d)$-SAT is the computational problem of determining whether a given $(k,d)$-formula is NAE-satisfiable.
\end{definition}

\begin{definition}
We define the function $f(k)$ as the maximum degree $d$ such that every $(k,d)$-formula is satisfiable. Similarly, we define $g(k)$ as the maximum degree $d$ such that every $(k,d)$-formula is NAE-satisfiable.
\end{definition}
\noindent
\textbf{The complexity jump.} The definition of $f(k)$ implies that $(k,f(k))$-SAT is trivial (as a decision problem). In~\cite{kratochvil1993one}, it was shown that while every $(k, f(k))$-formula is satisfiable, determining the satisfiability of a $(k, f(k)+1)$-formula is NP-complete. That is, there exists a \emph{complexity jump}, where the problem moves from being trivial to being NP-complete. Hence, determining the value of $f(k)$, as a function of $k$ is a problem of interest. 

The lopsided Lov\'{a}sz Local Lemma~\cite{erdos1991lopsided} was used by Gebauer, Szab\'{o}, and Tardos~\cite{gebauer2016local} to show that $f(k) \geq \lfloor\frac{2^{k+1}}{e(k+1)}\rfloor$. They also showed that $f(k) = \left(\frac{2}{e} + O\left(\frac{1}{\sqrt{k}}\right)\right)\frac{2^{k}}{k}$ by constructing unsatisfiable E$k$-CNF formulas with that degree. This determines the value of $f(k)$ asymptotically with respect to $k$.

We study the corresponding complexity jump for NAE-E$k$-SAT. We show that NAE-E$k$-SAT also exhibits a similar complexity jump, generalizing the result by~\cite{kratochvil1993one} (\Cref{thm:NAE-jump}). While this is a straightforward extension of \cite{kratochvil1993one}, we were not able to find this result in the literature prior to this. We also observe that $g(k) \geq \floor{\frac{f(k)}{2}}$ (\Cref{lem:g-lowerbnd}). We conjecture that $g(k)=\left(\frac{1}{e}+O\left(\frac{1}{\sqrt{k}}\right)\right)\frac{2^k}{k}$. 
\begin{theorem} \label{lem:g-lowerbnd}
    $g(k) \geq \floor{f(k)/2}\; .$
\end{theorem}

\begin{proof}
Let $\Phi$ be a $(k,d)$-formula, where $d = \left\lfloor\frac{f(k)}{2}\right\rfloor$. Construct a formula $\Phi^{\prime}$ as follows: For every clause $C$ in $\Phi$, we construct two clauses $C_{1}$ and $C_{2}$. $C_{1}$ is identical to $C$, and $C_{2}$ is obtained by negating all the literals in $C$. $\Phi^{\prime}$ is a $k$-CNF formula. Further, every variable in $\Phi$ appears in at most $2d$ clauses in $\Phi^{\prime}$. Hence, $\Phi^{\prime}$ is a $(k,2d)$-formula. Because $2d \leq f(k)$, $\Phi^{\prime}$ is satisfiable. Consider any satisfying assignment to $\Phi^{\prime}$. By definition, this assignment must satisfy both $C_{1}$ and $C_{2}$ for every clause $C$ in $\Phi$. Hence, it sets at least one literal to true and at least one literal to false in every clause of $\Phi$. This implies that the satisfying assignment is an NAE-satisfying assignment for $\Phi$. Thus, every $(k, d)$-formula is NAE-satisfiable, which means $g(k) \geq d = \left \lfloor\frac{f(k)}{2}\right\rfloor$.
\end{proof}

\begin{theorem} \label{thm:NAE-jump}
    Determining the NAE-satisfiability of $(k, g(k)+1)$ formulas is NP-hard. 
\end{theorem}

\begin{lemma} \label{lem:nae-unsat}
    Suppose there exists a $(k,d)$-formula $\Phi$ that is not NAE-satisfiable. Then, for any given set of $k-1$ variables $z_1, z_2, \dots, z_{k-1}$, we can construct a NAE-satisfiable $(k,d)$-formula $\Phi_{z_1, \dots, z_{k-1}}$ with the following property-- Any assignment that NAE-satisfies $\Phi_{z_1, \dots, z_{k-1}}$ assigns the same truth value to $z_1, z_2, \dots, z_{k-1}$.
\end{lemma}
\begin{proof}
    Without loss of generality, assume that $\Phi$ is a \emph{minimal} instance. That is, by removing any clause $C$, the remaining formula is NAE-satisfiable. Now, consider any clause $C$. Without loss of generality, we can assume the clause to be $z_1 \vee z_2 \vee \dots \vee z_k$ (we can relabel the variables and negate the literals accordingly). Let $\Phi=\Phi_1 \wedge C$, where $\Phi_1$ is the conjunction of the remaining clauses. Now, $\Phi_1$ is NAE-satisfiable, and every satisfying assignment to $\Phi_1$ assigns the same truth value to $z_1, z_2, \dots, z_{k-1}$, because $\Phi$ is not NAE-satisfiable. Now, replace $z_k$ in $C$ by a new variable $z'$ (which does not appear in any other clause) to get the formula $\Phi_{z_1,\dots, z_k}$ with the properties described above.  
\end{proof}

    \medskip \noindent
    \textbf{Proof of \Cref{thm:NAE-jump}.} We design a reduction from NAE-E$k$-SAT to NAE-$(k, g(k)+1)$-SAT. Let $\Psi$ be any E$k$-CNF formula. Let $y$ be any variable of $\Psi$ that appears in $m>g(k)+1$ clauses $C_1, \dots, C_m$. We show that we can replace $y$ with variables $y_1, \dots, y_m$ in the clauses $C_1, \dots, C_m$ to obtain a formula $\Tilde{\Psi}$ such that $\Tilde{\Psi}$ is NAE-satisfiable if and only if $\Psi$ is NAE-satisfiable. We now define $\Tilde{\Psi}$. We first define the formula $\Psi'$ by replacing the occurrence of $y$ in $C_1$ by a new variable $y_1$, $y$ in $C_2$ by a new variable $y_2$, etc. For each $i=1, \dots, m$, we also define the new variables $x_{i,1}, \dots, x_{i,k-2}$. We use \Cref{lem:nae-unsat} to define a formula $\Phi_{y_i,x_{i,1},\dots,x_{i,k-2}}$. For each $i$, we also define the chaining clauses $B_i=(y_i \vee \Bar{y}_{i+1}\vee x_{i,1} \dots x_{i,k-2})$. We now define the formula

    $$\Tilde{\Psi}:=\Psi' \bigwedge_{i=1}^m B_i \bigwedge_{i=1}^m\Phi_{y_i,x_{i,1},\dots,x_{i,k-2}} \; .$$

    Consider any assignment $\alpha$ to the variables. If $\alpha$ NAE-satisfies $\Tilde{\Psi}$, then $\alpha$ must NAE-satisfy $\Phi_{y_i,x_{i,1},\dots,x_{i,k-2}}$ for each $i$, and \Cref{lem:nae-unsat} implies that $y_i,x_{i,1},\dots,x_{i,k-2}$ are all assigned the same truth value. Suppose $\alpha$ assigns different truth values to $y_i$ and $y_{i+1}$. This implies that $\alpha$ assigns the same truth value to $y_i$ and $\Bar{y}_{i+1}$. Because we have already shown that $y_i,x_{i,1},\dots,x_{i,k-2}$ are assigned the same truth value, every variable in the clause $B_i$ has the same truth value, contradicting the assumption that $\alpha$ NAE-satisfies $\Tilde{\Psi}$. Hence, every assignment that NAE-satisfies $\Tilde{\Psi}$ has to assign the same truth value to the variables $y_i$ for each $i=1,\dots,m$, and $\Tilde{\Psi}$ is NAE-satisfiable if and only if $\Psi$ is NAE-satisfiable. We can repeat this process for each variable that appears in $>g(k)+1$ clauses to obtain a $(k,g(k)+1)$-formula that is NAE-satisfiable if and only if $\Psi$ is NAE-satisfiable.

\subsection{Algorithms for diverse bounded degree E$k$-SAT}
\label{sec:alg-disp-lll}

We first state the constructive Lovasz Local lemma and the Moser Tardos algorithm, rephrased in the language of CSPs. For a clause $C$, we use the notation $\mathsf{vbl}(C)$ to denote the variables that the clause $C$ is defined over.

\begin{theorem}[\cite{moser2010constructive}]
    \label{thm:LLL-CSP}
    Let $\Phi$ be a binary constraint satisfaction problem on $n$ variables and $m$ constraints $C_1,C_2, \dots, C_m$. For each $i \in [m]$, let $p_i$ denote the probability that constraint $C_i$ is violated by a uniformly random assignment to its variables. For two constraints $C_i, C_j$, we say that $C_i \sim C_j$ if and only if $\mathsf{vbl}(C_i) \bigcap \textsf{vbl}(C_j) \neq \emptyset$. For each $i \in [m]$, we define $\Gamma_i=\{j \neq i \in [m]: C_i \sim C_j\}$. Further, suppose there exist real numbers $x_1, \dots, x_m \in (0,1)$ such that for each $i \in [m]$,
    \begin{equation*}
        p_i \leq x_i \prod_{j \in \Gamma_i} (1-x_j) \; ,
    \end{equation*}
    Then, there exists an assignment to the $n$ variables that satisfies all the constraints in $\Phi$. Further, there exists a randomized algorithm that outputs such an assignment in expected time $O(m)$. 
\end{theorem}
In particular, if $\Phi$ is a E$k$-CNF formula and if each clause intersects with at most $2^k/e-1$ other clauses, $\Phi$ is satisfiable and a satisfying assignment can be found in polynomial time using the Moser-Tardos algorithm. The main result of this section is that we can also use the Moser-Tardos algorithm to get multiple diverse satisfying assignments to such E$k$-CNF formulas.

\begin{theorem} \label{thm:LLL-dispersion}
    Let $\Phi$ be a E$k$-CNF formula over $n$ variables and $n^{O(1)}$ clauses such that each clause intersects with at most $2^k/e-1$ other clauses. Then, for $s=2^{o(n)}$, and $\delta= H_2^{-1}\left(1-2/(k\ln(2))\right) \in (0,1)$, $\Phi$ has $s$ satisfying assignments $z_1^*, z_2^*, \dots, z_s^*$, such that $\sigma_{\text{min}}(z_1^*, z_2^*, \dots, z_s^*) \geq \delta n$. Furthermore, there exists a randomized algorithm that outputs such a set of satisfying assignments in expected time $poly(s,n)$.
\end{theorem}

\label{app:thm:LLL-dispersion}
\begin{proof}
    Fix any $t \le s$ and let $S=\{z_1^*,\dots,z_t^*\}\subseteq\Om$ be a set of satisfying
    assignments, and let $(C_1,\dots,C_m)$ be the clauses of $\f$. Define a CSP $\Psi_S$ on the
    same $n$ variables: introduce a single distance constraint $C_{\text{dist},S}$ with
    $\vbl(C_{\text{dist},S})=[n]$ and $C_{\text{dist},S}(z)=1$ iff $d_H(z,z_i^*)\ge\delta n$ for every
    $i\in[t]$, and let the constraints of $\Psi_S$ be $\{C_1,\dots,C_m\}\cup\{C_{\text{dist},S}\}$.
    Writing $\delta'=1-H_2(\delta)$, the probability $p_{\text{dist},S}$ that a uniformly random
    assignment violates $C_{\text{dist},S}$ satisfies
    \[
        p_{\text{dist},S}\;\le\;\frac{s\cdot 2^{H_2(\delta)n}}{2^n}\;=\;s\,2^{-\delta'n},
    \]
    since a Hamming ball of radius $\delta n$ has volume at most $2^{H_2(\delta)n}$.\footnote{$H_2(\cdot)$ is the binary entropy function and $2^{H_2(\eps)n}$ is the volume of a Hamming ball of radius $\eps n$.}
    If $\Psi_S$ satisfies the conditions of \Cref{thm:LLL-CSP}, then it is satisfiable, so $\f$ has a
    satisfying assignment at distance $\ge\delta n$ from every assignment in $S$; such an assignment
    can be found and added to $S$ in $\poly(n)$ time via Moser--Tardos.

    We verify the LLL conditions. Take the coefficients $x_1,\dots,x_{m+1}\in(0,1)$ of
    \Cref{thm:LLL-CSP} to be $x_i=e/2^k$ for $i\le m$, with $x_{m+1}$ assigned to $C_{\text{dist},S}$
    and set to $x_{m+1}=\frac{s\,2^{-\delta'n}}{(1-e/2^k)^m}$. The condition for $C_{\text{dist},S}$ holds
    by construction, since $p_{m+1}=s\,2^{-\delta'n}=x_{m+1}\prod_{i=1}^m(1-e/2^k)$. Write
    $D=\lfloor 2^k/e\rfloor-1$, so that $\frac{1}{D+1}=\frac{e}{2^k}$. Because $C_{\text{dist},S}$ is
    adjacent to every clause, for each $i\le m$,
    \[
        x_i\prod_{j\in\Gamma_i}(1-x_j)
        =\frac{1}{D+1}\Big(1-\tfrac{1}{D+1}\Big)^{D}\big(1-x_{m+1}\big)
        =\frac{1}{D+1}\cdot\frac{1}{(1+1/D)^D}\cdot\big(1-x_{m+1}\big),
    \]
    and it suffices to bound this below by $1/2^k$.

    There is a constant $\eps=\eps(k)>0$ with $(1+1/D)^{-D}\ge\frac1e(1+\eps)$; combined with
    $\frac{1}{D+1}=\frac{e}{2^k}$ this gives $\frac{1}{D+1}\cdot\frac{1}{(1+1/D)^D}\ge\frac{1+\eps}{2^k}$.
    Hence it suffices to ensure $1-x_{m+1}\ge\frac{1}{1+\eps}$, i.e.
    $x_{m+1}\le\eps':=\frac{\eps}{1+\eps}$. Since every clause meets at most $D$ others, every variable
    lies in at most $D+1$ clauses, so $m\le\frac{(D+1)n}{k}=\frac{2^k n}{ek}$, and with
    $\delta=H_2^{-1}\!\big(1-\tfrac{2}{k\ln 2}\big)$ (so $\delta'=\tfrac{2}{k\ln 2}$),
    \[
        x_{m+1}=\frac{s\,2^{-\delta'n}}{(1-e/2^k)^m}
        \;\le\;\frac{s\,2^{-\delta'n}}{(1-e/2^k)^{2^k n/(ek)}}
        \;=\;s\,A^{-n/k},\qquad A:=e^{2}\big(1-e/2^k\big)^{2^k/e}.
    \]
    For $k\ge 3$ we have $A>1$ (and $A\to e$ as $k\to\infty$), so $x_{m+1}=s\,A^{-n/k}$.

    \textbf{The algorithm.} The bound $x_{m+1}\le\eps'$, and hence the LLL
    precondition for $\Psi_S$, is true when $n\ge n_0:=\frac{k}{\ln A}\ln(s/\eps')=\Theta(\log s)$.
    The construction therefore has two regimes. For $n\ge n_0$ we build $S^*$ inductively: find
    $z_1^*$ by Moser--Tardos, and at step $i+1$ run Moser--Tardos on $\Psi_{S_i}$ (where
    $S_i=\{z_1^*,\dots,z_i^*\}$) to obtain $z_{i+1}^*$ with $d_H(z_{i+1}^*,z_j^*)\ge\delta n$ for all
    $j\le i$; after $s$ steps this yields $S^*=\{z_1^*,\dots,z_s^*\}$ with pairwise distance $\ge\delta n$.
    For $n<n_0$ we have $2^n<(s/\eps')^{k\ln 2/\ln A}=\poly(s)$, so the same separated set is obtained
    by exhaustive search in $\poly(s,n)$ time. Either way the procedure runs in $\poly(s,n)$ time; in
    particular, when $s=2^{o(n)}$ we have $n_0=o(n)$, so the exhaustive branch is invoked for only
    finitely many $n$.
\end{proof}

\bibliographystyle{alpha}
\bibliography{bib2doi}

@inproceedings{galvez2025computing,
  title={Computing diverse and nice triangulations},
  author={G{\'a}lvez, Waldo and Goswami, Mayank and Merino, Arturo and Park, GiBeom and Tsai, Meng-Tsung},
  booktitle={International Symposium on Fundamentals of Computation Theory},
  pages={180--193},
  year={2025},
  organization={Springer}
}

@article{galvez2025framework,
  title={A framework for the design of efficient diversification algorithms to np-hard problems},
  author={G{\'a}lvez, Waldo and Goswami, Mayank and Merino, Arturo and Park, GiBeom and Tsai, Meng-Tsung and Verdugo, Victor},
  journal={arXiv preprint arXiv:2501.12261},
  year={2025}
}

@article{kratochvil1993one,
  title = {One more occurrence of variables makes satisfiability jump from trivial to NP-complete},
  author = {Kratochv{\'\i}l, Jan and Savick{\`y}, Petr and Tuza, Zsolt},
  journal = {SIAM Journal on Computing},
  volume = {22},
  number = {1},
  pages = {203--210},
  year = {1993},
  publisher = {SIAM},
  timestamp = {Sat, 27 May 2017 01:00:00 +0200},
  biburl = {https://dblp.org/rec/journals/siamcomp/KratochvilST93.bib},
  bibsource = {dblp computer science bibliography, https://dblp.org},
  doi = {10.1137/0222015},
  _bib2doi_selected = {dblp:/rec/journals/siamcomp/KratochvilST93.bib},
  _bib2doi_confirmed = {true},
}

@inproceedings{erdos1975problems,
  title = {Problems and results on 3-chromatic hypergraphs and some related questions},
  author = {Erd{\H{o}}s, Paul and Lov{\'a}sz, L{\'a}szl{\'o}},
  booktitle = {Infinite and finite sets},
  volume = {10},
  pages = {609--627},
  year = {1975},
  organization = {North-Holland, Amsterdam},
}

@article{gebauer2016local,
  title = {The local lemma is asymptotically tight for {SAT}},
  author = {Gebauer, Heidi and Szab{\'o}, Tibor and Tardos, G{\'a}bor},
  journal = {Journal of the ACM (JACM)},
  volume = {63},
  number = {5},
  pages = {1--32},
  year = {2016},
  publisher = {ACM New York, NY, USA},
  timestamp = {Sun, 19 Jan 2025 00:00:00 +0100},
  biburl = {https://dblp.org/rec/journals/jacm/GebauerST16.bib},
  bibsource = {dblp computer science bibliography, https://dblp.org},
  doi = {10.1145/2975386},
  _bib2doi_selected = {dblp:/rec/journals/jacm/GebauerST16.bib},
  _bib2doi_confirmed = {true},
}

@article{moser2010constructive,
  title = {A constructive proof of the general Lov{\'a}sz local lemma},
  author = {Moser, Robin A and Tardos, G{\'a}bor},
  journal = {Journal of the ACM (JACM)},
  volume = {57},
  number = {2},
  pages = {1--15},
  year = {2010},
  publisher = {ACM New York, NY, USA},
  timestamp = {Sun, 19 Jan 2025 00:00:00 +0100},
  biburl = {https://dblp.org/rec/journals/jacm/MoserT10.bib},
  bibsource = {dblp computer science bibliography, https://dblp.org},
  doi = {10.1145/1667053.1667060},
  _bib2doi_selected = {dblp:/rec/journals/jacm/MoserT10.bib},
  _bib2doi_confirmed = {true},
}

@inproceedings{bhangale2015simultaneous,
  title = {Simultaneous approximation of constraint satisfaction problems},
  author = {Bhangale, Amey and Kopparty, Swastik and Sachdeva, Sushant},
  booktitle = {International Colloquium on Automata, Languages, and Programming},
  pages = {193--205},
  year = {2015},
  organization = {Springer},
  timestamp = {Sat, 19 Oct 2019 01:00:00 +0200},
  biburl = {https://dblp.org/rec/conf/icalp/BhangaleKS15.bib},
  bibsource = {dblp computer science bibliography, https://dblp.org},
  doi = {10.1007/978-3-662-47672-7_16},
  _bib2doi_selected = {dblp:/rec/conf/icalp/BhangaleKS15.bib},
  _bib2doi_confirmed = {true},
}

@article{erdos1991lopsided,
  title = {Lopsided Lov{\'{a}}sz Local Lemma and Latin transversals},
  author = {Paul Erd{\"{o}}s and Joel Spencer},
  journal = {Discret. Appl. Math.},
  volume = {30},
  number = {2-3},
  pages = {151--154},
  year = {1991},
  timestamp = {Tue, 01 Jun 2021 01:00:00 +0200},
  biburl = {https://dblp.org/rec/journals/dam/ErdosS91.bib},
  bibsource = {dblp computer science bibliography, https://dblp.org},
  doi = {10.1016/0166-218X(91)90040-4},
  _bib2doi_selected = {dblp:/rec/journals/dam/ErdosS91.bib},
  _bib2doi_confirmed = {true},
  _bib2doi_finished = {true},
}

@article{baste2022diversity,
  title = {Diversity of solutions: An exploration through the lens of fixed-parameter tractability theory},
  author = {Baste, Julien and Fellows, Michael R. and Jaffke, Lars and Masa{\v{r}}{\'\i}k, Tom{\'a}{\v{s}} and de Oliveira Oliveira, Mateus and Philip, Geevarghese and Rosamond, Frances A.},
  journal = {Artificial Intelligence},
  volume = {303},
  pages = {103644},
  year = {2022},
  publisher = {Elsevier},
  timestamp = {Fri, 21 Jan 2022 00:00:00 +0100},
  biburl = {https://dblp.org/rec/journals/ai/BasteFJMOPR22.bib},
  bibsource = {dblp computer science bibliography, https://dblp.org},
  doi = {10.1016/j.artint.2021.103644},
  _bib2doi_selected = {dblp:/rec/journals/ai/BasteFJMOPR22.bib},
  _bib2doi_confirmed = {true},
}

@inproceedings{austrin2025algorithms,
  title = {Algorithms for the Diverse-k-{SAT} Problem: The Geometry of Satisfying Assignments},
  author = {Per Austrin and Ioana O. Bercea and Mayank Goswami and Nutan Limaye and Adarsh Srinivasan},
  booktitle = {52nd International Colloquium on Automata, Languages, and Programming, {ICALP} 2025, Aarhus, Denmark, July 8-11, 2025},
  year = {2025},
  organization = {Schloss Dagstuhl-Leibniz-Zentrum f{\"u}r Informatik},
  timestamp = {Mon, 30 Jun 2025 01:00:00 +0200},
  biburl = {https://dblp.org/rec/conf/icalp/AustrinB0LS25.bib},
  bibsource = {dblp computer science bibliography, https://dblp.org},
  doi = {10.4230/LIPIcs.ICALP.2025.14},
  publisher = {Schloss Dagstuhl - Leibniz-Zentrum f{\"{u}}r Informatik},
  volume = {334},
  pages = {14:1--14:17},
  editor = {Keren Censor{-}Hillel and Fabrizio Grandoni and Jo{\"{e}}l Ouaknine and Gabriele Puppis},
  series = {LIPIcs},
  _bib2doi_selected = {dblp:/rec/conf/icalp/AustrinB0LS25.bib},
  _bib2doi_confirmed = {true},
  _bib2doi_finished = {true},
}

@inproceedings{nadel2011generating,
  title = {Generating Diverse Solutions in {SAT}},
  author = {Alexander Nadel},
  booktitle = {Theory and Applications of Satisfiability Testing - {SAT} 2011 - 14th International Conference, {SAT} 2011, Ann Arbor, MI, USA, June 19-22, 2011. Proceedings},
  pages = {287--301},
  year = {2011},
  organization = {IEEE},
  timestamp = {Sun, 25 Oct 2020 01:00:00 +0200},
  biburl = {https://dblp.org/rec/conf/sat/Nadel11.bib},
  bibsource = {dblp computer science bibliography, https://dblp.org},
  doi = {10.1007/978-3-642-21581-0_23},
  publisher = {Springer},
  volume = {6695},
  editor = {Karem A. Sakallah and Laurent Simon},
  series = {Lecture Notes in Computer Science},
  _bib2doi_selected = {dblp:/rec/conf/sat/Nadel11.bib},
  _bib2doi_confirmed = {true},
  _bib2doi_finished = {true},
}

@article{gonzalez1985clustering,
  title = {Clustering to minimize the maximum intercluster distance},
  author = {Gonzalez, Teofilo F},
  journal = {Theoretical Computer Science},
  volume = {38},
  pages = {293--306},
  year = {1985},
  publisher = {Elsevier},
  timestamp = {Wed, 17 Feb 2021 00:00:00 +0100},
  biburl = {https://dblp.org/rec/journals/tcs/Gonzalez85.bib},
  bibsource = {dblp computer science bibliography, https://dblp.org},
  doi = {10.1016/0304-3975(85)90224-5},
  _bib2doi_selected = {dblp:/rec/journals/tcs/Gonzalez85.bib},
  _bib2doi_confirmed = {true},
}

@inproceedings{agbaria2010sat,
  title = {{SAT}-based semiformal verification of hardware},
  author = {Agbaria, Sabih and Carmi, Dan and Cohen, Orly and Korchemny, Dmitry and Lifshits, Michael and Nadel, Alexander},
  booktitle = {Formal Methods in Computer Aided Design},
  pages = {25--32},
  year = {2010},
  organization = {IEEE},
  timestamp = {Mon, 09 Aug 2021 01:00:00 +0200},
  biburl = {https://dblp.org/rec/conf/fmcad/AgbariaCCKLN10.bib},
  bibsource = {dblp computer science bibliography, https://dblp.org},
  url = {https://ieeexplore.ieee.org/document/5770929/},
  _bib2doi_selected = {dblp:/rec/conf/fmcad/AgbariaCCKLN10.bib},
  _bib2doi_confirmed = {true},
}

@inproceedings{hebrard2005finding,
  title = {Finding diverse and similar solutions in constraint programming},
  author = {Hebrard, Emmanuel and Hnich, Brahim and O'Sullivan, Barry and Walsh, Toby},
  booktitle = {Proceedings of the 20th National Conference on Artificial Intelligence (AAAI)},
  pages = {372--377},
  year = {2005},
  timestamp = {Tue, 05 Sep 2023 01:00:00 +0200},
  biburl = {https://dblp.org/rec/conf/aaai/HebrardHOW05.bib},
  bibsource = {dblp computer science bibliography, https://dblp.org},
  url = {http://www.aaai.org/Library/AAAI/2005/aaai05-059.php},
  _bib2doi_selected = {dblp:/rec/conf/aaai/HebrardHOW05.bib},
  _bib2doi_confirmed = {true},
}

@inproceedings{fomin2020diversity,
  title = {Diversity of Solutions: An Exploration Through the Lens of Fixed-Parameter Tractability Theory},
  author = {Julien Baste and Michael R. Fellows and Lars Jaffke and Tom{\'{a}}s Masar{\'{\i}}k and Mateus de Oliveira Oliveira and Geevarghese Philip and Frances A. Rosamond},
  booktitle = {Proceedings of the Twenty-Ninth International Joint Conference on Artificial Intelligence, {IJCAI} 2020},
  pages = {1119--1125},
  year = {2020},
  timestamp = {Sat, 09 Apr 2022 01:00:00 +0200},
  biburl = {https://dblp.org/rec/conf/ijcai/BasteFJMOPR20.bib},
  bibsource = {dblp computer science bibliography, https://dblp.org},
  doi = {10.24963/ijcai.2020/156},
  publisher = {ijcai.org},
  editor = {Christian Bessiere},
  _bib2doi_selected = {dblp:/rec/conf/ijcai/BasteFJMOPR20.bib},
  _bib2doi_confirmed = {true},
  _bib2doi_finished = {true},
}

@inproceedings{liang2025diversat,
  title = {Diversat: a novel and effective local search algorithm for diverse {SAT} problem},
  author = {Liang, Jiaxin and Zhou, Junping and Yin, Minghao},
  booktitle = {Proceedings of the AAAI Conference on Artificial Intelligence},
  volume = {39},
  number = {11},
  pages = {11290--11298},
  year = {2025},
  timestamp = {Thu, 17 Apr 2025 01:00:00 +0200},
  biburl = {https://dblp.org/rec/conf/aaai/LiangZY25.bib},
  bibsource = {dblp computer science bibliography, https://dblp.org},
  doi = {10.1609/aaai.v39i11.33228},
  _bib2doi_selected = {dblp:/rec/conf/aaai/LiangZY25.bib},
  _bib2doi_confirmed = {true},
}

@article{gima2024computing,
  title = {Computing diverse pair of solutions for tractable {SAT}},
  author = {Gima, Tatsuya and Iwamasa, Yuni and Kobayashi, Yasuaki and Kurita, Kazuhiro and Otachi, Yota and Saito, Rin},
  journal = {arXiv preprint arXiv:2412.04016},
  year = {2024},
  timestamp = {Tue, 14 Jan 2025 00:00:00 +0100},
  biburl = {https://dblp.org/rec/journals/corr/abs-2412-04016.bib},
  bibsource = {dblp computer science bibliography, https://dblp.org},
  doi = {10.48550/arXiv.2412.04016},
  _bib2doi_selected = {dblp:/rec/journals/corr/abs-2412-04016.bib},
  _bib2doi_confirmed = {true},
  _bib2doi_finished = {true},
}

@inproceedings{hanaka2023framework,
  title = {A framework to design approximation algorithms for finding diverse solutions in combinatorial problems},
  author = {Hanaka, Tesshu and Kiyomi, Masashi and Kobayashi, Yasuaki and Kobayashi, Yusuke and Kurita, Kazuhiro and Otachi, Yota},
  booktitle = {Proceedings of the AAAI Conference on Artificial Intelligence},
  volume = {37},
  number = {4},
  pages = {3968--3976},
  year = {2023},
  timestamp = {Mon, 04 Sep 2023 01:00:00 +0200},
  biburl = {https://dblp.org/rec/conf/aaai/HanakaK0KKO23.bib},
  bibsource = {dblp computer science bibliography, https://dblp.org},
  doi = {10.1609/aaai.v37i4.25511},
  _bib2doi_selected = {dblp:/rec/conf/aaai/HanakaK0KKO23.bib},
  _bib2doi_confirmed = {true},
}

@inproceedings{cevallos2016local,
  title = {Local search for max-sum diversification},
  author = {Alfonso Cevallos and Friedrich Eisenbrand and Rico Zenklusen},
  booktitle = {Proceedings of the Twenty-Eighth Annual {ACM-SIAM} Symposium on Discrete Algorithms, {SODA} 2017, Barcelona, Spain, Hotel Porta Fira, January 16-19},
  pages = {130--142},
  year = {2017},
  organization = {SIAM},
  timestamp = {Tue, 02 Feb 2021 00:00:00 +0100},
  biburl = {https://dblp.org/rec/conf/soda/CevallosEZ17.bib},
  bibsource = {dblp computer science bibliography, https://dblp.org},
  doi = {10.1137/1.9781611974782.9},
  url = {https://doi.org/10.1137/1.9781611974782.9},
  editor = {Philip N. Klein},
  publisher = {{SIAM}},
  _bib2doi_selected = {dblp:/rec/conf/soda/CevallosEZ17.bib},
  _bib2doi_confirmed = {true},
  _bib2doi_finished = {true},
}

@inproceedings{zhou2023lsdtkms,
  title = {{LS-DTKMS:} {A} Local Search Algorithm for Diversified Top-k {MaxSAT} Problem},
  author = {Junping Zhou and Jiaxin Liang and Minghao Yin and Bo He},
  booktitle = {26th International Conference on Theory and Applications of Satisfiability Testing, {SAT} 2023, Alghero, Italy, July 4-8, 2023},
  volume = {271},
  pages = {29:1--29:16},
  year = {2023},
  organization = {Schloss Dagstuhl-Leibniz-Zentrum f{\"u}r Informatik},
  timestamp = {Thu, 10 Aug 2023 01:00:00 +0200},
  biburl = {https://dblp.org/rec/conf/sat/ZhouLYH23.bib},
  bibsource = {dblp computer science bibliography, https://dblp.org},
  doi = {10.4230/LIPIcs.SAT.2023.29},
  publisher = {Schloss Dagstuhl - Leibniz-Zentrum f{\"{u}}r Informatik},
  editor = {Meena Mahajan and Friedrich Slivovsky},
  series = {LIPIcs},
  _bib2doi_selected = {dblp:/rec/conf/sat/ZhouLYH23.bib},
  _bib2doi_confirmed = {true},
  _bib2doi_finished = {true},
}

@article{hastad_optimal_2001,
  title = {Some optimal inapproximability results},
  volume = {48},
  issn = {0004-5411, 1557-735X},
  url = {https://dl.acm.org/doi/10.1145/502090.502098},
  doi = {10.1145/502090.502098},
  abstract = {We prove optimal, up to an arbitrary ε {\textgreater} 0, inapproximability results for Max-E k -Sat for k ≥ 3, maximizing the number of satisfied linear equations in an over-determined system of linear equations modulo a prime p and Set Splitting. As a consequence of these results we get improved lower bounds for the efficient approximability of many optimization problems studied previously. In particular, for Max-E2-Sat, Max-Cut, Max-di-Cut, and Vertex cover.},
  language = {en},
  number = {4},
  urldate = {2025-11-17},
  journal = {Journal of the ACM},
  author = {Håstad, Johan},
  month = {jul},
  year = {2001},
  pages = {798--859},
  timestamp = {Wed, 14 Nov 2018 00:00:00 +0100},
  biburl = {https://dblp.org/rec/journals/jacm/Hastad01.bib},
  bibsource = {dblp computer science bibliography, https://dblp.org},
  _bib2doi_selected = {dblp:/rec/journals/jacm/Hastad01.bib},
  _bib2doi_confirmed = {true},
}

@article{hastad_advantage_2004,
  title = {On the advantage over a random assignment},
  volume = {25},
  copyright = {http://onlinelibrary.wiley.com/termsAndConditions\#vor},
  issn = {1042-9832, 1098-2418},
  url = {https://onlinelibrary.wiley.com/doi/10.1002/rsa.20031},
  doi = {10.1002/rsa.20031},
  abstract = {Abstract We initiate the study of a new measure of approximation. This measure compares the performance of an approximation algorithm to the random assignment algorithm. This is a useful measure for optimization problems where the random assignment algorithm is known to give essentially the best possible polynomial time approximation. In this paper, we focus on this measure for the optimization problems Max‐Lin‐2 in which we need to maximize the number of satisfied linear equations in a system of linear equations modulo 2, and Max‐ k ‐Lin‐2, a special case of the above problem in which each equation has at most k variables. The main techniques we use, in our approximation algorithms and inapproximability results for this measure, are from Fourier analysis and derandomization. © 2004 Wiley Periodicals, Inc. Random Struct. Alg., 2004},
  language = {en},
  number = {2},
  urldate = {2025-11-30},
  journal = {Random Structures \& Algorithms},
  author = {Håstad, Johan and Venkatesh, S.},
  month = {sep},
  year = {2004},
  pages = {117--149},
  timestamp = {Wed, 14 Nov 2018 00:00:00 +0100},
  biburl = {https://dblp.org/rec/journals/rsa/HastadV04.bib},
  bibsource = {dblp computer science bibliography, https://dblp.org},
  _bib2doi_selected = {dblp:/rec/journals/rsa/HastadV04.bib},
  _bib2doi_confirmed = {true},
}

@article{trevisan1998parallel,
  title = {Parallel approximation algorithms by positive linear programming},
  author = {Trevisan, Luca},
  journal = {Algorithmica},
  volume = {21},
  number = {1},
  pages = {72--88},
  year = {1998},
  publisher = {Springer},
  timestamp = {Tue, 25 Feb 2025 00:00:00 +0100},
  biburl = {https://dblp.org/rec/journals/algorithmica/Trevisan98.bib},
  bibsource = {dblp computer science bibliography, https://dblp.org},
  doi = {10.1007/PL00009209},
  _bib2doi_selected = {dblp:/rec/journals/algorithmica/Trevisan98.bib},
  _bib2doi_confirmed = {true},
}

@inproceedings{hast2005approximating,
  title = {Approximating - Outperforming a Random Assignment with Almost a Linear Factor},
  author = {Gustav Hast},
  booktitle = {Automata, Languages and Programming, 32nd International Colloquium, {ICALP} 2005, Lisbon, Portugal, July 11-15, 2005, Proceedings},
  pages = {956--968},
  year = {2005},
  organization = {Springer},
  timestamp = {Tue, 23 May 2017 01:00:00 +0200},
  biburl = {https://dblp.org/rec/conf/icalp/Hast05.bib},
  bibsource = {dblp computer science bibliography, https://dblp.org},
  doi = {10.1007/11523468_77},
  publisher = {Springer},
  volume = {3580},
  editor = {Lu{\'{\i}}s Caires and Giuseppe F. Italiano and Lu{\'{\i}}s Monteiro and Catuscia Palamidessi and Moti Yung},
  series = {Lecture Notes in Computer Science},
  _bib2doi_selected = {dblp:/rec/conf/icalp/Hast05.bib},
  _bib2doi_confirmed = {true},
  _bib2doi_finished = {true},
}

@article{charikar2009near,
  title = {Near-optimal algorithms for maximum constraint satisfaction problems},
  author = {Charikar, Moses and Makarychev, Konstantin and Makarychev, Yury},
  journal = {ACM Transactions on Algorithms (TALG)},
  volume = {5},
  number = {3},
  pages = {1--14},
  year = {2009},
  publisher = {ACM New York, NY, USA},
  timestamp = {Sat, 30 Sep 2023 01:00:00 +0200},
  biburl = {https://dblp.org/rec/journals/talg/CharikarMM09.bib},
  bibsource = {dblp computer science bibliography, https://dblp.org},
  doi = {10.1145/1541885.1541893},
  _bib2doi_selected = {dblp:/rec/journals/talg/CharikarMM09.bib},
  _bib2doi_confirmed = {true},
}

@inproceedings{samorodnitsky2006gowers,
  title = {Gowers uniformity, influence of variables, and PCPs},
  author = {Samorodnitsky, Alex and Trevisan, Luca},
  booktitle = {Proceedings of the thirty-eighth annual ACM symposium on Theory of Computing},
  pages = {11--20},
  year = {2006},
  timestamp = {Tue, 25 Feb 2025 00:00:00 +0100},
  biburl = {https://dblp.org/rec/conf/stoc/SamorodnitskyT06.bib},
  bibsource = {dblp computer science bibliography, https://dblp.org},
  doi = {10.1145/1132516.1132519},
  _bib2doi_selected = {dblp:/rec/conf/stoc/SamorodnitskyT06.bib},
  _bib2doi_confirmed = {true},
}

@article{cevallos2019improved,
  title = {An improved analysis of local search for max-sum diversification},
  author = {Cevallos, Alfonso and Eisenbrand, Friedrich and Zenklusen, Rico},
  journal = {Mathematics of Operations Research},
  volume = {44},
  number = {4},
  pages = {1494--1509},
  year = {2019},
  publisher = {INFORMS},
  timestamp = {Thu, 19 Dec 2019 00:00:00 +0100},
  biburl = {https://dblp.org/rec/journals/mor/CevallosEZ19.bib},
  bibsource = {dblp computer science bibliography, https://dblp.org},
  doi = {10.1287/moor.2018.0982},
  _bib2doi_selected = {dblp:/rec/journals/mor/CevallosEZ19.bib},
  _bib2doi_confirmed = {true},
}

@book{odonnell2014analysis,
  title = {Analysis of Boolean Functions},
  author = {O'Donnell, Ryan},
  year = {2014},
  publisher = {Cambridge University Press},
  note = {Theorem 9.23},
  timestamp = {Mon, 01 Sep 2014 01:00:00 +0200},
  biburl = {https://dblp.org/rec/books/daglib/0033652.bib},
  bibsource = {dblp computer science bibliography, https://dblp.org},
  isbn = {978-1-10-703832-5},
  _bib2doi_selected = {dblp:/rec/books/daglib/0033652.bib},
  _bib2doi_confirmed = {true},
}

@book{macwilliams1977theory,
  title = {The Theory of Error-Correcting Codes},
  author = {MacWilliams, Florence Jessie and Sloane, Neil James Alexander},
  volume = {16},
  year = {1977},
  publisher = {Elsevier},
  _bib2doi_finished = {true},
}

@book{guruswami2019essential,
  title = {Essential Coding Theory},
  author = {Guruswami, Venkatesan and Rudra, Atri and Sudan, Madhu},
  year = {2019},
  publisher = {Draft available online},
  note = {\url{https://cse.buffalo.edu/~atri/courses/coding-theory/book/}},
  _bib2doi_finished = {true},
}

@article{beck1991algorithmic,
  title = {An algorithmic approach to the {Lov{\'a}sz Local Lemma} {I}},
  author = {Beck, Joseph},
  journal = {Random Structures \& Algorithms},
  volume = {2},
  number = {4},
  pages = {343--365},
  year = {1991},
  publisher = {Wiley Online Library},
  timestamp = {Wed, 14 Nov 2018 00:00:00 +0100},
  biburl = {https://dblp.org/rec/journals/rsa/Beck91a.bib},
  bibsource = {dblp computer science bibliography, https://dblp.org},
  doi = {10.1002/rsa.3240020402},
  _bib2doi_selected = {dblp:/rec/journals/rsa/Beck91a.bib},
  _bib2doi_confirmed = {true},
}

@article{drabik2024finding,
  title = {Finding Diverse Solutions Parameterized by Cliquewidth},
  author = {Karolina Drabik and Tom{\'{a}}s Masar{\'{\i}}k},
  booktitle = {Proceedings of the AAAI Conference on Artificial Intelligence},
  year = {2024},
  timestamp = {Mon, 24 Jun 2024 01:00:00 +0200},
  biburl = {https://dblp.org/rec/journals/corr/abs-2405-20931.bib},
  bibsource = {dblp computer science bibliography, https://dblp.org},
  doi = {10.48550/arXiv.2405.20931},
  journal = {CoRR},
  volume = {abs/2405.20931},
  url = {https://doi.org/10.48550/arXiv.2405.20931},
  eprinttype = {arXiv},
  eprint = {2405.20931},
  _bib2doi_selected = {dblp:/rec/journals/corr/abs-2405-20931.bib},
  _bib2doi_confirmed = {true},
  _bib2doi_finished = {true},
}

@article{iwamasa2025general,
  title = {A General Framework for Finding Diverse Solutions via Network Flow and Its Applications},
  author = {Yuni Iwamasa and Tomoki Matsuda and Shunya Morihira and Hanna Sumita},
  booktitle = {36th International Symposium on Algorithms and Computation (ISAAC 2025)},
  volume = {abs/2504.17633},
  pages = {41:1--41:17},
  year = {2025},
  organization = {Schloss Dagstuhl-Leibniz-Zentrum f{\"u}r Informatik},
  timestamp = {Fri, 23 May 2025 01:00:00 +0200},
  biburl = {https://dblp.org/rec/journals/corr/abs-2504-17633.bib},
  bibsource = {dblp computer science bibliography, https://dblp.org},
  doi = {10.48550/arXiv.2504.17633},
  eprinttype = {arXiv},
  eprint = {2504.17633},
  journal = {CoRR},
  _bib2doi_selected = {dblp:/rec/journals/corr/abs-2504-17633.bib},
  _bib2doi_confirmed = {true},
  _bib2doi_finished = {true},
}

@article{justesen2003class,
  title={Class of constructive asymptotically good algebraic codes},
  author={Justesen, J{\o}rn},
  journal={IEEE Transactions on information theory},
  volume={18},
  number={5},
  pages={652--656},
  year={2003},
  publisher={IEEE}
}

@inproceedings{paturi1997satisfiability,
  title={Satisfiability coding lemma},
  author={Paturi, Ramamohan and Pudl{\'a}k, Pavel and Zane, Francis},
  booktitle={Proceedings 38th Annual Symposium on Foundations of Computer Science},
  pages={566--574},
  year={1997},
  organization={IEEE}
}

@inproceedings{schoning1999probabilistic,
  title={A probabilistic algorithm for k-{SAT} and constraint satisfaction problems},
  author={Schoning, T},
  booktitle={40th Annual Symposium on Foundations of Computer Science (Cat. No. 99CB37039)},
  pages={410--414},
  year={1999},
  organization={IEEE}
}

@article{scheder2024ppsz,
  title={PPSZ is better than you think},
  author={Scheder, Dominik},
  journal={TheoretiCS},
  volume={3},
  year={2024},
  publisher={Episciences. org}
}

@article{paturi2005improved,
  title={An improved exponential-time algorithm for k-{SAT}},
  author={Paturi, Ramamohan and Pudl{\'a}k, Pavel and Saks, Michael E and Zane, Francis},
  journal={Journal of the ACM (JACM)},
  volume={52},
  number={3},
  pages={337--364},
  year={2005},
  publisher={ACM New York, NY, USA}
}

@article{austrin2024algorithmsarxiv,
  title={Algorithms for the {Diverse}-k-{SAT} problem: the geometry of satisfying assignments},
  author={Austrin, Per and Bercea, Ioana O and Goswami, Mayank and Limaye, Nutan and Srinivasan, Adarsh},
  journal={arXiv preprint arXiv:2408.03465},
  year={2024}
}

@article{misra2025parameterized,
  title={On the parameterized complexity of diverse {SAT}},
  author={Misra, Neeldhara and Mittal, Harshil and Rai, Ashutosh},
  journal={Theoretical Computer Science},
  pages={115653},
  year={2025},
  publisher={Elsevier}
}

@inproceedings{zheng2025exact,
  title={Exact Approaches for the Diverse Satisfiability Problem},
  author={Zheng, Zhifei and Cherif, Sami and S{\'a} Shibasaki, Rui and Li, Chu-Min and Zhang, Jialu},
  booktitle={European Conference on Logics in Artificial Intelligence},
  pages={207--224},
  year={2025},
  organization={Springer}
}

@article{crescenzi_hamming_2002,
    title = {On the {Hamming} distance of constraint satisfaction problems},
    volume = {288},
    issn = {03043975},
    url = {https://linkinghub.elsevier.com/retrieve/pii/S0304397501001463},
    doi = {10.1016/S0304-3975(01)00146-3},
    language = {en},
    number = {1},
    urldate = {2023-06-15},
    journal = {Theoretical Computer Science},
    author = {Crescenzi, P. and Rossi, G.},
    month = oct,
    year = {2002},
    pages = {85--100},
}

\end{document}